%% file: Dhruva_DGA_Main.tex
\RequirePackage{fix-cm}
\documentclass[smallextended]{svjour3}       
\smartqed  
\usepackage{graphicx}

\usepackage{url,amsmath,amsfonts,amssymb, bm,xcolor,bbm}
\usepackage{mathrsfs,mathtools,enumerate}

\newtheorem{assumption}{Assumption}

\newcommand{\R}{{\mathbb R}}

\newcommand{\E}{{\mathbb E}}
\newcommand{\X}{{\mathcal X}}
\newcommand{\U}{{\mathcal U}}
\newcommand{\Y}{{\mathcal Y}}
\renewcommand{\P}{{\mathcal P}}
\newcommand{\Py}{{\mathbb P}}
\newcommand{\C}{\mathcal{C}}
\newcommand{\G}{\mathcal{G}}

\renewcommand{\H}{\mathcal{H}}
\newcommand{\vct}[1]{{#1}}

\newcommand{\rv}[1]{{{#1}}}
\newcommand{\gm}[1]{{\mathscr{#1}}}

\begin{document}

\title{Upper and Lower Values in Zero-sum Stochastic Games with Asymmetric Information\thanks{Preliminary version of this paper appears in the proceedings of the 58th Conference on Decision and Control (CDC), 2019 \cite{kartik2019stochastic}.}
}

\titlerunning{Zero-sum Stochastic Games with Asymmetric Information}        

\author{Dhruva Kartik         \and
        Ashutosh Nayyar 
}


\institute{Dhruva Kartik \at
              \email{mokhasun@usc.edu}           
           \and
          Ashutosh Nayyar \at
              \email{ashutosh.nayyar@usc.edu}\\
              \\
              Ming Hsieh Department
of Electrical and Computer Engineering, University of Southern California, Los Angeles,
CA, 90007 USA
}

\date{Received: date / Accepted: date}

\maketitle

\begin{abstract}
\input{abstract.tex}
\end{abstract}

\section{Introduction}
\label{intro}
\input{intro.tex}

\section{Problem Formulation}\label{sec:probform}
\input{probform.tex}

\section{Virtual Game $\gm{G}_v$}\label{virtual}
\input{virtgame.tex}

\section{Expanded Virtual Game $\gm{G}_{e}$ with Prescription History}\label{sec:expanded}
\input{expgame.tex}

\input{expgame_part2.tex}

\section{Games with Complete Information on One Side and Partial Information on the Other}\label{sec:incomp}
\input{oneside.tex}



\section{Conclusion}\label{sec:conc}
\input{conc.tex}

\begin{appendix}
\input{append_cdc.tex}
\end{appendix}

\bibliographystyle{spmpsci}      
\bibliography{refs}   


\end{document}

%% file: abstract.tex
A general model for zero-sum stochastic games with asymmetric information is considered. In this model, each player's information at each time can be divided into a common information part and a private information part. Under certain conditions on the evolution of the common and private information, a dynamic programming characterization of the value of the game (if it exists) is presented. If the value of the zero-sum game does not exist, then the dynamic program provides bounds on the upper and lower values of the game. This dynamic program is then used for a class of zero-sum stochastic games with complete information on one side and partial information on the other, that is, games where one player has complete information about state, actions and observation history while the other player may only have partial information about the state and action history. For such games, it is shown that the value exists and can be characterized using the dynamic program. {It is further shown that for this class of games, the dynamic program can be used to compute an equilibrium strategy for the more informed player in which the player selects its action using its private information and the common information belief.}

%% file: intro.tex
%
%
%
%

Zero-sum games have been widely used as a model of strategic decision-making in the presence of adversaries. Such decision-making scenarios arise in a range of domains including (i) security of cyber-physical and infrastructure systems such as  the power grid and  water networks in the presence of cyber or physical attacks \cite{washburn1995two,amin2015game,amin2012cyber,wu2018securing,zhu2015game,shelar2016security}, (ii) cyber-security of networked computing and communication systems \cite{washburn1995two,alpcan2010network}, (iii) designing anti-poaching measures \cite{fang2015security,fang2017paws,bondi2019using}, (iv) military operations in the presence of hostile agents \cite{haywood1954military} and, (v) competitive markets and geopolitical interactions \cite{morrow1994game,aumann1995repeated}. In many cases, the adversarial interactions occur over time in a dynamic and uncertain environment. Zero-sum stochastic games provide a useful model for these situations. In these games, two players may jointly control the evolution of the state of a stochastic dynamic system with one player trying to minimize the total cost while the other trying to maximize it. In stochastic games with symmetric information, all players have the same information about the state and action histories. Such games have been extensively studied in the literature in both zero-sum and nonzero-sum settings \cite{fudenberg1991game,filar2012competitive,basar1999dynamic}. In many situations of interest, however, the players may have different information about the state and action histories. A potential attacker of a cyber-physical system for example, may not have the same information as the defender; adversaries in a battlefield may have different information about the surroundings and about each other. The focus of this paper is on such \emph{asymmetric information} settings. 


We adopt a model of asymmetric information that was originally developed for decentralized stochastic control \cite{nayyar2013decentralized}. This model partitions each player's information at each time into a common information part and a private information part. The common information at time $t$ is known to all players at that time and at all times in the future. In addition to the common information, each player may have some private information. It has been noted in the existing literature that this model subsumes a wide range of information structures \cite{nayyar2013decentralized,nayyar2014common}. 

In our model, it may be the case that no player knows the current state of the underlying stochastic system perfectly. Further,  since each player may have some private information, one player's information is not necessarily included in the other player's information.  The partial observability of the state, the asymmetry of information and the fact that each player may have some private information complicates the characterization and computation of the equilibrium cost (value) and equilibrium strategies. We provide two results for this general model of zero-sum stochastic game with asymmetric information: (i) If the game has a Nash equilibrium in behavioral strategies, then our result provides dynamic programming based characterizations of the value of the game. Each step of these programs involves a min-max (or a max-min) problem over the space of \emph{prescriptions} which are functions from players' private information to actions. (ii) If the game does not have a Nash equilibrium, then our dynamic programs provide a lower bound on the upper value of the game and an upper bound on the lower value of the game.

{We then specialize our model to the case where (i) one player (say, the attacker) has \emph{partial information} on the system state and the other player's (say, the defender's)  actions; (ii) the other player (i.e. the defender) has complete information, that is, it knows the system state as well as the attacker's information. We allow both players to control the system state. Our model can be viewed as a generalization of the models in \cite{renault2006value,renault2012value,li2014lp,zheng2013decomposition}. We first show that a Nash equilibrium exists in our model and thus the upper and lower values are equal. This allows us to characterize the value of the game using our dynamic programming approach. We also describe some structural properties of the value functions in the dynamic program that could be leveraged for computational efficiency. Further, we find a sufficient statistic for the more informed player, i.e, we show that there exists a Nash equilibrium where the more informed player plays a \emph{common information belief} based strategy \cite{nayyar2014common}, \cite{ouyang2017dynamic} and that our dynamic programming approach can be used to compute this strategy.}

\subsection{Related Work}\label{sec:related}
\begin{enumerate}
\item \emph{Stochastic games of symmetric information:} In this stochastic game model, the players have access to the same information. Thus, at any time $t$, each player has no uncertainty regarding other players' information and makes a decision anticipating the other players' strategies. Such games of symmetric information have been extensively studied in the literature \cite{fudenberg1991game,filar2012competitive,basar1999dynamic}. Because of this symmetry, players' shared information (or a function of it) can be treated as a state and utilized to decompose a dynamic game into simpler single-stage Bayesian games. These single-stage games can then be solved in a backward inductive manner to obtain the value and Nash equilibria (if any exist). In this paper, we focus on models in which players have different information and thus the methodology described above for symmetric information games is not directly applicable to our model.
\item  \emph{Zero-sum games with limited information on one side:} Stochastic zero-sum games with complete information on one side and limited information on the other have been investigated before with varying degrees of generality. In \cite{mertens2015repeated,renault2006value}, the state evolution was uncontrolled, in \cite{renault2012value,li2014lp}, the state could only be controlled by the more-informed player, and in \cite{zheng2013decomposition}, the state could be controlled by both players. In all these works (except \cite{renault2012value}), both players' actions are commonly observed and, the less-informed player has no state information. A related model in which the system is uncontrolled and players' actions are commonly observed but both players may have asymmetric state information has been investigated in \cite{gensbittel2015value}. In our general model in Section \ref{sec:probform}, both players may have imperfect information about the system and the other player's actions, and both players may control the system state. We also consider a specialized model in Section \ref{sec:incomp} where (i) the defender has complete information while the attacker has partial information on the state history and the defender's action history; (ii) both players may influence the state evolution. These two features differentiate our work from prior work mentioned above.
\item \emph{Stochastic games of asymmetric information with strategy-independent common information beliefs: } In \cite{nayyar2014common}, a common information based dynamic program was developed for finding Nash equilibria in general (i.e. not necessarily zero-sum) stochastic games of asymmetric information. The key idea in this approach is to first convert the game of asymmetric information into a virtual game of symmetric information. This virtual game of symmetric information is then solved using a common information based dynamic program. However, this approach relies on an assumption on the players' information (see Assumption 2 in \cite{nayyar2014common}). This assumption holds only for certain classes of information structures and may not necessarily be true for the asymmetric information games described in Sections \ref{sec:probform} and \ref{sec:incomp}.
\item \emph{Common information based perfect Bayesian equilibria in stochastic games of asymmetric information:} Authors in \cite{ouyang2017dynamic} consider a stochastic game model in which the system state can be decomposed into a public state that is commonly observed by all players and a private state that is privately observed by each player. In this model, all the players' past actions are commonly observed and, additionally, an imperfect version of players' private state may be disclosed to all the players at each time. {A special case of this model has been considered in \cite{vasal2019systematic}.} For the models in \cite{ouyang2017dynamic} and \cite{vasal2019systematic}, the authors provide characterizations of perfect Bayesian equilibria under some assumptions on the evolution of players' private state. 
In this paper, we focus only on two-player zero-sum games. However, the system dynamics and the information structure in our model are more general than those in the model of \cite{ouyang2017dynamic,vasal2019systematic}. For instance, unlike in  \cite{ouyang2017dynamic,vasal2019systematic}, players' actions may not be fully observed in our model. Further, the solutions in \cite{ouyang2017dynamic,vasal2019systematic} rely on strong existence assumptions that may not be true in general.

\end{enumerate}

Our work is most closely related to \cite{nayyar2017information} and \cite{nayyar2014common} . We follow the approach in \cite{nayyar2017information} and build on its results. The system model in \cite{nayyar2017information} conformed to a specific structure, that is, the system state could be decomposed into three components: a public state that is commonly observed (perhaps partially) and a privately observed component for each player. The model in our paper is substantially more general than in \cite{nayyar2017information}. Another major restriction in \cite{nayyar2017information} was that the players were allowed to play only pure strategies. In this paper, we allow the players to play behavioral strategies. 
Our model is similar to \cite{nayyar2014common} but {we do not make the critical assumption made in \cite{nayyar2014common} that the common information based beliefs be strategy-independent (see Assumption 2 of \cite{nayyar2014common}). Removing this assumption  makes our model much more widely applicable than the model in \cite{nayyar2014common}.}

\subsection{Contributions}
The main contributions of our paper are:
\begin{enumerate}
\item We consider a general stochastic zero-sum game model in which the players select their actions using different information. For this general model, we provide a dynamic programming characterization of the value of the zero-sum game, if it exists. If the value does not exist, then our characterization provides bounds on the upper and lower values of the zero-sum game.

\item We then specialize our model to the case in which the defender has complete information and the attacker may have partial information. For this model, we show that the value of the zero-sum game exists and that our dynamic program characterizes the value of this game.

\item For the game in Section \ref{sec:incomp}, we show that our dynamic program can be used to find an equilibrium strategy for the more-informed player such that its behavioral action is a function of its private information and the common information based belief.

\item For the specialized model in Section \ref{sec:incomp}, we prove that the value functions in our dynamic program satisfy some structural properties that can be leveraged to make the dynamic program computationally more tractable. One such property is that the value functions are piecewise linear and convex in the common information belief.

\end{enumerate}

\subsection{Notation}\label{notation}
Random variables/vectors are denoted by upper case letters, their realizations by the corresponding lower case letters. In general, subscripts are used as time index while superscripts are used to index decision-making agents. For time indices $t_1\leq t_2$, $\rv{X}_{t_1:t_2}$ (resp. $g_{t_1:t_2}$) is the short hand notation for the variables $(\rv{X}_{t_1},\rv{X}_{t_1+1},...,\rv{X}_{t_2})$ (resp.  functions $(g_{t_1},\dots,g_{t_2})$). Similarly, $\rv{X}^{1:2}$ is the short hand notation for the collection of variables $(\rv{X}^1,\rv{X}^2)$.
Operators $\Py(\cdot)$ and $\E[\cdot]$ denote the probability of an event, and the expectation of a random variable respectively.
For random variables/vectors $X$ and $Y$, $\Py(\cdot | \rv{Y}=y)$, $\E[\rv{X}| \rv{Y}=y]$ and $\Py(\rv{X} = x \mid \rv{Y} = y)$ are denoted by $\Py(\cdot | y)$, $\E[\rv{X}|y]$ and $\Py(x \mid y)$, respectively. 
For a strategy $g$, we use $\Py^g(\cdot)$ (resp. $\E^g[\cdot]$) to indicate that the probability (resp. expectation) depends on the choice of $g$. For any finite set $\mathcal{A}$, $\Delta\mathcal{A}$ denotes the probability simplex over the set $\mathcal{A}$.

\subsection{Organization}
The rest of the paper is organized as follows.  We formulate the game in Section \ref{sec:probform} and construct a virtual game with symmetric information in Section \ref{virtual}. In Section \ref{sec:expanded}, we construct an expanded virtual game and use it provide a dynamic programming characterization of the value. In Section \ref{sec:incomp}, we analyze the model with complete information on one side and partial information on the other. We conclude the paper in Section \ref{sec:conc}. Proofs of key results are provided in the appendices.

%% file: probform.tex
Consider a dynamic system with two players. The system operates in discrete time over a horizon $T$. Let $\rv{X}_t \in \X_t$ be the state of the system at time $t$, and let $\rv{U}_t^i \in \U_t^i$ be the action of player $i$ at time $t$, where $i = 1,2$. The state of the system evolves in a controlled Markovian manner as
\begin{align}
\label{statevol}\rv{X}_{t+1} = f_t(\rv{X}_t, \rv{U}_t^1,\rv{U}_t^2,\rv{W}_t^s),
\end{align}
where $\rv{W}_t^s$ is the system noise. There are two observation processes $\rv{Y}_t^1 \in \Y_t^1$ and $\rv{Y}_t^2 \in \Y_t^2$ given as
\begin{align}
\label{obseq}\rv{Y}_t^i = h_t^i(\rv{X}_t, \rv{U}_{t-1}^1,\rv{U}_{t-1}^2,\rv{W}_t^i), \; i = 1,2,
\end{align}
where $\rv{W}_t^1$ and $\rv{W}_t^2$ are observation noises. We assume that the sets $\X_t, \U_t^i$ and $\Y_t^i$ are finite for all $i$ and $t$. Further,  the random variables $\rv{X}_1, \rv{W}_t^s, \rv{W}_t^i$ (referred to as \emph{the primitive random variables}) can take finitely many values and are mutually independent.

\subsection{Information Structure}\label{infostruct}
The collection of variables (i.e. observations, actions) available to player $i$ at time $t$ is denoted by $\rv{I}_t^i$. 
$\rv{I}_t^i$ is a subset of all observations until time $t$ and actions until $t-1$, i.e,   $ \rv{I}_t^i \subseteq \{\rv{Y}^{1:2}_{1:t}, \rv{U}^{1:2}_{1:t-1}\} $. The set of all possible realizations of $\rv{I}_t^i$ is denoted by $\mathcal{I}^i_t$. 

Information $\rv{I}_t^i$ can be decomposed into \emph{private} and \emph{common} information, i.e. $\rv{I}_t^i = \rv{C}_t \cup \rv{P}_t^i$. Common information $\rv{C}_t$ is the set of variables known to both players at time $t$ while variables in the private information $\rv{P}_t^i$ are known only to player $i$. Let $\mathcal{C}_t$ be the set of all realizations of common information at time $t$ and let $\mathcal{P}_t^i$ be the set of all realizations of private information for player $i$ at time $t$.  We make the following assumption on the evolution of common and private information. This is similar to Assumption 1 of  \cite{nayyar2014common}\footnote{Note that we do not impose Assumption 2 of \cite{nayyar2014common}.}. 

\begin{assumption}\label{infevolve}
The evolution of common and private information available to the players is as follows:
\begin{enumerate}
\item The common information $\rv{C}_t$ is increasing with time, i.e. $\rv{C}_t \subset \rv{C}_{t+1}$. Let $\rv{Z}_{t+1} := \rv{C}_{t+1}\setminus \rv{C}_t$ be the increment in common information. Thus, $\rv{C}_{t+1} = \{\rv{C}_t,\rv{Z}_{t+1}\}$. Furthermore,
\begin{align}
\label{commonevol}\rv{Z}_{t+1} = \zeta_{t+1}(\rv{P}_t^{1:2},\rv{U}_t^{1:2},\rv{Y}_{t+1}^{1:2}),
\end{align}
where $\zeta_{t+1}$ is a fixed transformation.
\item The private information evolves as
\begin{align}
\label{privevol}\rv{P}^i_{t+1} = \xi_{t+1}^i(\rv{P}_t^i,\rv{U}_t^i,\rv{Y}_{t+1}^i),
\end{align}
where $\xi_{t+1}^i$ is a fixed transformation.
\end{enumerate}
\end{assumption}
As noted in \cite{nayyar2013decentralized} and \cite{nayyar2014common}, a number of information structures satisfy the above assumption.
We briefly mention a few below:
\begin{enumerate}
\item  \emph{No common information:} Consider the case where each player only has access to its own observations and actions, i.e., $I^i_t=\{Y^i_{1:t},U^i_{1:t-1}\}, i = 1,2$. In this case, there is no common information, i.e., $C_t=\varnothing$. It is easy to verify that Assumption \ref{infevolve} is valid in this case.
\item \emph{No private information:} Consider the case where all observations and actions are public, i.e., $I^i_t = \{Y_{1:t}^{1:2},U_{1:t-1}^{1:2}\}$. In this case, players do not have any private information, i.e., $P_t^i = \varnothing$. Once again, Assumption \ref{infevolve} is true.
\item \emph{Information structure in \cite{vasal2019systematic}:} In this model, player $i$ has a private state $X_t^i$, $i = 1,2$. Player $i$ knows its private state and both players' actions are commonly observed. This information structure can be seen as a special case of our model in the following manner: let the state $X_t := (X_t^1,X_t^2)$ and the observation processes $Y_t^i = X_t^i$ for $i = 1,2$. Define the information sets at time $t$ as $I_t^i = \{Y_{1:t}^i, U_{1:t-1}^{1:2}\}.$ In this case, $C_t = \{U_{1:t-1}^1,U_{1:t-1}^2\}$ and $P_t^i = Y_{1:t}^i$. Clearly, this information structure satisfies Assumption \ref{infevolve}. Similarly, information structures in \cite{ouyang2017dynamic} and \cite{nayyar2017information} can also be seen as special cases of our model.
\item \emph{Full state information on one side and quantized state information on the other:} Consider the model in which player 1 knows the state $X_t$ and player 2 sees a quantized version of $X_t$. That is, $Y_t^1 = X_t$ and $Y_t^2 = q(X_t)$ where $q$ is an arbitrary function. Both players' actions are commonly observed. Since player 1 knows the state $X_t$ and $Y_t^2$ is a deterministic function of the state, player 1 also knows player 2's observation $Y_t^2$. {Thus, in this case, $I_t^1 = \{Y_{1:t}^{1:2},U_{1:t-1}^{1:2}\}$ and $I_t^2 = \{Y_{1:t}^{2},U_{1:t-1}^{1:2}\}$. Therefore, $C_t = \{Y_{1:t}^{2},U_{1:t-1}^{1:2}\}$, $P_t^1 = Y_{1:t}^{1}$ and $P_t^2 = \varnothing$. Clearly, this model satisfies Assumption \ref{infevolve}.}
\end{enumerate}

\subsection{Strategies and Values}
Players can use any information available to them to select their actions and we allow behavioral strategies for both players. Thus, player $i$ chooses a distribution $\delta\rv{U}_t^i$ over its action space using a \emph{control law} $g_t^i:  \mathcal{I}_t^i \rightarrow \Delta\mathcal{U}_t^i$, i.e.
\begin{equation}
\delta\rv{U}_t^i = g_t^i(\rv{I}_t^i) = g_t^i(\rv{P}_t^i,\rv{C}_t). \label{eq:stg}
\end{equation}
Player $i$'s action at time $t$ is randomly chosen from $\U_t^i$ according to the distribution $\delta\rv{U}_t^i$. We will at times refer to $\delta \rv{U}^i_t$ as player $i$'s \emph{behavioral action} at time $t$. It will be helpful for our analysis to explicitly describe the randomization procedure used by the players. To do so, we assume that player $i$ has access to i.i.d. random variables $\rv{K	}^i_{1:T}$ that are uniformly distributed over the interval $(0,1]$. The variables $\rv{K}^1_{1:T}, \rv{K}^2_{1:T}$ are independent of each other and of the primitive random variables. Further, player $i$ has access to a mechanism $\kappa$ that takes as input $\rv{K}^i_t$ and a distribution over $\mathcal{U}^i_t$ and generates a random action with the input distribution. Thus, player $i$'s action at time $t$ can be written as $\rv{U}_t^i = \kappa(g_t^i(\rv{I}_t^i),\rv{K}_t^i).$
\begin{remark}
One choice of the mechanism $\kappa$ can be described as follows: Suppose  $\mathcal{U}^i_t = \{1,2,..n\}$ and the input distribution is $(p_1,...p_n)$. We can \emph{partition} the interval $(0,1]$ into $n$ intervals $(a_i,b_i]$ such that the length of $i$th interval is $b_i -a_i = p_i$. Then, $U^i_t=k$ if $\rv{K}_t^i \in (a_k,b_k]$ for $k=1,\ldots,n$.
\end{remark}

The collection of control laws $\vct{g}^i = (g_1^i,\dots,g_T^i)$ is referred to as the \emph{control strategy} of player $i$, and the pair of control strategies $(\vct{g}^1,\vct{g}^2)$ is referred to as a \emph{strategy profile}. Let the set of all possible control strategies for player $i$ be $\G^i$.

The total expected cost associated with a strategy profile $(\vct{g}^1,\vct{g}^2)$ is
\begin{align}
J(\vct{g}^1,\vct{g}^2):=\E^{(\vct{g}^1,\vct{g}^2)}\left[\sum_{t=1}^T c_t(\rv{X}_t,\rv{U}_t^{1},\rv{U}_t^2)\right], \label{eq:totalcost}
\end{align}
where $c_t:\X_t \times \U_t^1 \times \U_t^2 \rightarrow \R$ is the cost function at time $t$.
Player 1 wants to minimize the total expected cost, while Player 2 wants to maximize it. We refer to this zero-sum game as Game $\gm{G}$. 
\begin{definition}
The upper value of the game $\gm{G}$ is defined as
\begin{align}
S^u(\gm{G}) := \inf_{g^1 \in \mathcal{G}^1}\sup_{g^2 \in \mathcal{G}^2} J(g^1,g^2).
\end{align}
The lower value of the game $\gm{G}$ is defined as
\begin{align}
S^l(\gm{G}) := \sup_{g^2 \in \mathcal{G}^2}\inf_{g^1 \in \mathcal{G}^1} J(g^1, g^2).
\end{align}
If the upper and lower values are the same,  they are referred to as  the value of the game and denoted by $S(\gm{G})$.   
\end{definition}
A Nash equilibrium of the zero-sum game $\gm{G}$ is a strategy profile $(\vct{g}^{1*},\vct{g}^{2*})$ such that for every $\vct{g}^1 \in \G^1$ and $\vct{g}^2 \in \G^2$, we have
\begin{align}
J(\vct{g}^{1*},\vct{g}^{2}) \leq J(\vct{g}^{1*},\vct{g}^{2*}) \leq J(\vct{g}^{1},\vct{g}^{2*}).
\end{align}
Nash equilibria in zero-sum games satisfy the following property \cite{osborne1994course}.
\begin{proposition}
If a Nash equilibrium in Game $\gm{G}$ exists, then for every Nash equilibrium $(\vct{g}^{1*},\vct{g}^{2*})$ in Game $\gm{G}$, we have
\begin{align}
J(\vct{g}^{1*},\vct{g}^{2*}) = S^l(\gm{G}) = S^u(\gm{G}) = S(\gm{G}).
\end{align}
\end{proposition}

\begin{remark}\label{kuhnremark}
Note that the existence of a Nash equilibrium is not guaranteed in general. However, if players have perfect recall, i.e. 
\begin{equation}
\{\rv{U}^i_{1:t-1}\} \cup \rv{I}_{t-1}^i \subseteq \rv{I}_t^i
\end{equation} for every $i$ and $t$, then the existence of a behavioral strategy equilibrium is guaranteed by {Kuhn's theorem \cite{maschler2013game}}.
\end{remark}

The objective of this work is to characterize the upper and lower values $S^u(\gm{G})$ and $S^l(\gm{G})$ of Game $\gm{G}$. To this end, we will define a virtual game $\gm{G}_v$ and an ``expanded'' virtual game $\gm{G}_e$. These virtual games will be used to  obtain bounds on the upper and lower values of the original game $\gm{G}$.


%% file: virtgame.tex
The virtual game $\gm{G}_v$ is constructed using the methodology in \cite{nayyar2014common}. This game involves the same set of primitive random variables as in Game $\gm{G}$. The two players of game $\gm{G}$ are replaced by two virtual players in $\gm{G}_v$. The virtual players operate as follows. At each time $t$, virtual player $i$ selects a function $\Gamma^i_t$ that maps private information $P^i_t$ to a distribution $\delta \rv{U}_t^i$ over the space $\mathcal{U}_t^i$. We refer to these functions as \emph{prescriptions}.  Let $\mathcal{B}_t^i$ be the set of all possible prescriptions for virtual player $i$ at time $t$ (i.e.  $\mathcal{B}_t^i$ is the set of all mappings from $\mathcal{P}_t^i$ to $\Delta \mathcal{U}_t^i$).

Once the virtual players select their prescriptions, the action $U^i_t$ is randomly generated according to distribution $\Gamma_t^i(\rv{P}_t^i)$. More precisely, the system dynamics for this game are given by:
\begin{align}
\label{virdyn1}\rv{X}_{t+1} &= f_t(\rv{X}_t,\rv{U}_t^{1:2},\rv{W}_t^s)\\
\label{virdyn2}\rv{P}^i_{t+1} &= \xi_{t+1}^i(\rv{P}_t^i,\rv{U}_t^i,\rv{Y}_{t+1}^i) & i = 1,2,\\
\label{virdyn3}\rv{Y}_{t+1}^i &= h_{t+1}^i(\rv{X}_{t+1},\rv{U}_t^{1:2},\rv{W}_{t+1}^i) & i = 1,2,\\
\label{virdyn4}\rv{U}_t^i &= \kappa(\Gamma_t^i(\rv{P}_t^i),\rv{K}_t^i) & i = 1,2,\\
\label{virdyn5}\rv{Z}_{t+1} &= \zeta_{t+1}(\rv{P}_t^{1:2},\rv{U}_t^{1:2},\rv{Y}_{t+1}^{1:2}),
\end{align}
where the functions $f_t,h_t^i$, $\xi_t^i,\kappa$ and $\zeta_{t}$  are the same as in $\gm{G}$.



In the virtual game, virtual players use the common information $\rv{C}_t$ to select their prescriptions at time $t$. The $i$th virtual player selects its prescription according to a control law $\chi_t^i$, i.e. $\Gamma_t^i = \chi_t^i(\rv{C}_t)$.
 For virtual player $i$, the collection of control laws over the entire time horizon $\vct{\chi}^i = (\chi_1^i,\dots,\chi_T^i)$ is referred to as its control strategy. Let   $\mathcal{H}_t^i$ be the set of all possible control laws for virtual player $i$ at time $t$  and let  $\mathcal{H}^i$ be the set of all possible control strategies for virtual player $i$, i.e. $\mathcal{H}^i = \mathcal{H}_1^i \times \dots \times \mathcal{H}_T^i$. 
The total cost associated with the game for a strategy profile $(\vct{\chi}^1,\vct{\chi}^2)$ is
\begin{align}
\mathcal{J}(\vct{\chi}^1,\vct{\chi}^2)=\E^{(\vct{\chi}^1,\vct{\chi}^2)}\left[\sum_{t=1}^T c_t(\rv{X}_t,\rv{U}_t^{1},\rv{U}_t^2)\right], \label{eq:virtualJv}
\end{align}
where the function $c_t$ is the same as in Game $\gm{G}$.

The following lemma establishes a connection between the original game $\gm{G}$ and the virtual game $\gm{G}_v$ constructed above. 

\begin{lemma}\label{virtlemma}
Let $S^u(\gm{G}_v)$ and $S^l(\gm{G}_v)$ be, respectively, the upper and lower values of the virtual game $\gm{G}_v$. Then, \[S^l(\gm{G}) = S^l(\gm{G}_v) ~~\mbox{and}~~ S^u(\gm{G}) = S^u(\gm{G}_v).\] 
Consequently, if a Nash equilibrium exists in the original game $\gm{G}$, then $S(\gm{G})=S^l(\gm{G}_v) = S^u(\gm{G}_v).$ 
\end{lemma}
\begin{proof}
See Appendix \ref{virtlemmaproof}.\qed
\end{proof}


The authors in \cite{nayyar2014common} use the virtual game to find equilibrium costs and strategies for a stochastic dynamic game of asymmetric information. However, the methodology in \cite{nayyar2014common} is applicable \emph{only  under the assumption that the posterior beliefs on state $\rv{X}_t$ and private information $\rv{P}^{1,2}_t$ given the common information $C_t$ do not depend on the strategy profile being used} (see Assumption 2 in \cite{nayyar2014common}). We will refer to this assumption as the \emph{strategy-independent beliefs} (SIB) assumption. As pointed out in \cite{nayyar2014common}, the SIB assumption is satisfied by some special system models and information structures but is not true for general stochastic dynamic games.  A simple example which does not satisfy the SIB assumption is the following delayed sharing information structure \cite{nayyar2010optimal}: Consider game $\gm{G}$ with common information $C_t = 
\{Y^{1,2}_{1:t-2},U^{1,2}_{1:t-2} \}$ and $P^i_t = \{Y^i_t, Y^i_{t-1}, U^i_{t-1}\}$. Another example of a game where the SIB assumption fails is presented in Section \ref{sec:incomp}.
 
Thus, we are faced with the following situation: if our zero-sum game satisfies the SIB assumption, we can adopt the results in \cite{nayyar2014common} to find equilibrium costs (i.e. the value) of our game. However, if the zero-sum game does not satisfy the SIB assumption, then the methodology of \cite{nayyar2014common} is inapplicable. In the next section, we will develop a methodology to bound the upper and lower values of the zero-sum game $\gm{G}$ even when the game does not satisfy the SIB assumption. 

%
%
%
%

%% file: expgame.tex

In order to circumvent the SIB assumption, we now construct an expanded virtual game $\gm{G}_e$ by increasing the amount of information available to virtual players in game $\gm{G}_v$. In this new game $\gm{G}_e$, the state dynamics, observation processes, primitive random variables and cost function are all the same as in the  game $\gm{G}_v$. The only difference is in the information used by the virtual players to select their prescriptions. The virtual players now have access to the common information $\rv{C}_t$ \textit{as well as} all the past prescriptions of both players, i.e., $\Gamma_{1:t-1}^{1:2}$.   Virtual player $i$ selects its prescription at time $t$  using a control law $\tilde{\chi}_t^i$, i.e, $\Gamma_t^i = \tilde{\chi}_t^i(\rv{C}_t,\Gamma_{1:t-1}^{1:2}).$
Let $\tilde{\mathcal{H}}_t^i$ be the set of all such control laws at time $t$ for virtual player $i$. 
$\tilde{\mathcal{H}}^i  := \tilde{\mathcal{H}}_1^i \times \dots \times \tilde{\mathcal{H}}_T^i$ is the set of all control strategies for player $i$. 
The total cost associated with the game for a strategy profile $(\tilde{\vct{\chi}}^1,\tilde{\vct{\chi}}^2)$ is
\begin{align}
{\mathcal{J}}(\tilde{\vct{\chi}}^1,\tilde{\vct{\chi}}^2)=\E^{(\tilde{\vct{\chi}}^1,\tilde{\vct{\chi}}^2)}\left[\sum_{t=1}^T c_t(\rv{X}_t,\rv{U}_t^{1},\rv{U}_t^2)\right]. \label{eq:virtualJ}
\end{align}

\begin{remark}
Note that any strategy $\chi^i \in \H^i$ is equivalent to the strategy $\tilde{\chi}^i \in \tilde{\H}^i$ that satisfies the following condition: for each time $t$ and for each realization of common information $\vct{c}_t \in \C_t$,
\begin{align}
\tilde{\chi}_t^i(\vct{c}_t,\gamma_{1:t-1}^{1:2}) = \chi_t^i(\vct{c}_t) \quad \forall\;\gamma_{1:t-1}^{1:2} \in \mathcal{B}_{1:t-1}^{1:2}.
\end{align}
Hence, with slight abuse of notation, we can say that the strategy space $\H^i$ in the virtual game $\gm{G}_v$ is a subset of the strategy space $\tilde{\H}^i$ in the expanded game $\gm{G}_e$. {For this reason, the function $\mathcal{J}$ in \eqref{eq:virtualJ} can be thought of as an extension of the  function $\mathcal{J}$ in \eqref{eq:virtualJv}.    }
\end{remark}

\begin{remark}
Expansion of information structures has been used in prior work to find equilibrium costs/strategies. See, for example, \cite{bacsar1981saddle} which studies a linear stochastic differential game where both players have a common noisy observation of the state.
\end{remark}

%% file: expgame_part2.tex
\subsection{Upper and Lower Values of Games $\gm{G}_v$ and $\gm{G}_e$}

We will now establish the relationship between the upper and lower values of the expanded game $\gm{G}_e$ and the virtual game $\gm{G}_v$. To do so, we define the following mappings between the strategies in games $\gm{G}_v$ and $\gm{G}_{e}$. 

\begin{definition}\label{def:rho}
Let $\varrho^i: \tilde{\H}^1 \times  \tilde{\H}^2 \rightarrow \H^i$ be an operator that maps a strategy profile  $(\vct{\tilde{\chi}}^1,\vct{\tilde{\chi}}^2)$ in virtual game $\gm{G}_e$ to a strategy  $ \vct{{\chi}}^i$ for virtual player $i$ in  game $\gm{G}_v$ as follows: For $t=1,2,\ldots,T,$
\begin{align}
\chi_t^i(\vct{c}_t) := \tilde{\chi}_t^i(\vct{c}_t, \tilde{\gamma}_{1:t-1}^{1:2}),
\end{align}
where  $\tilde{\gamma}_s^j = \tilde{\chi}_s^j(\vct{c}_s,\tilde{\gamma}_{1:s-1}^{1:2})$ for every $1\leq s \leq t-1$ and $j = 1,2$.
We denote the ordered pair $(\varrho^1,\varrho^2)$ by $\varrho$.
\end{definition}
The mapping $\varrho$ is defined in such a way that the strategy profile  $(\vct{\tilde{\chi}}^1,\vct{\tilde{\chi}}^2)$  and the strategy profile  $\varrho(\vct{\tilde{\chi}}^1,\vct{\tilde{\chi}}^2)$ induce identical dynamics in the respective games $\gm{G}_e$ and $\gm{G}_v$. 




\begin{lemma}\label{evolequi}
Let $(\vct{{\chi}}^1,\vct{{\chi}}^2)$ and $(\vct{\tilde{\chi}}^1,\vct{\tilde{\chi}}^2)$ be  strategy profiles for games $\gm{G}_v$ and $\gm{G}_{e}$, such that $\vct{\chi}^i = \varrho^i(\vct{\tilde{\chi}}^1,\vct{\tilde{\chi}}^2)$, $i=1,2$. Then, 
\begin{align}
\mathcal{J}(\vct{{\chi}}^1,\vct{{\chi}}^2) = {\mathcal{J}}(\vct{\tilde{\chi}}^1,\vct{\tilde{\chi}}^2).
\end{align}
\end{lemma}
\begin{proof}
See Appendix \ref{evolequiproof}.\qed
\end{proof}

The following theorem connects the upper and lower values of the two virtual games and the original game.
\begin{theorem}\label{origvirt}
The lower and upper values of the three games defined above satisfy the following:
\begin{align*}
S^l(\gm{G}) =S^l(\gm{G}_v) \leq S^l(\gm{G}_e) \leq S^u(\gm{G}_e) \leq  S^u(\gm{G}_v) = S^u(\gm{G}).
\end{align*}
Consequently, if a Nash equilibrium exists in the original game $\gm{G}$, then $S(\gm{G}) = S^l(\gm{G}_e) = S^u(\gm{G}_e)$.
\end{theorem}
\begin{proof}
See Appendix \ref{origvirtproof}.\qed
\end{proof}

Using Theorem \ref{origvirt}, we can obtain bounds on the upper and lower values of the original game by computing the upper and lower values of the expanded game $\gm{G}_e$.

\subsection{The Dynamic Programming Characterization}\label{dpsec}
We now describe a methodology for finding the upper and lower values of the expanded game $\gm{G}_e$. Suppose the virtual players are using the strategy profile $(\tilde{\vct{\chi}}^{1}, \tilde{\vct{\chi}}^{2})$  in the expanded game $\gm{G}_e$. Let $\Pi_t$ be the virtual players' belief on the  state and private information based on their information  in game $\gm{G}_e$. Thus, $\Pi_t$ is 
 defined as 
\begin{align*}
\Pi_t(\vct{x}_t,\vct{p}_t^{1:2}) := \Py^{(\tilde{\vct{\chi}}^{1}, \tilde{\vct{\chi}}^{2})}(\rv{X}_t = \vct{x}_t,\rv{P}_t^{1:2} = \vct{p}_t^{1:2} \mid \rv{C}_t,\Gamma_{1:t-1}^{1:2}), \; \forall x_t,p_t^1,p_t^2.
\end{align*}
We refer to $\Pi_t$ as the \emph{common information belief} (CIB). {$\Pi_t$ takes values in the set $\mathcal{S}_t := \Delta(\mathcal{X}_t \times \mathcal{P}^1_t \times \mathcal{P}^2_t)$.}

\begin{definition}
{Given a belief $\pi$ on the state and private informations at time $t$ and mappings $\gamma^i, i=1,2,$ from $\mathcal{P}^i_t$ to $\Delta\mathcal{U}^i_t$, we define $\gamma^i(p_t^i ; u)$ as the probability assigned to action $u$ under the probability distribution $\gamma^i(p_t^i)$. 
Also, define
\begin{align}
\label{tildec}\tilde{c}_t(\pi,\gamma^1,\gamma^2) :=&\sum_{\vct{x}_t,\vct{p}_t^{1:2},\vct{u}_t^{1:2}}c_t(\vct{x}_t, \vct{u}_t^1,\vct{u}_t^2) \pi(\vct{x}_t,\vct{p}_t^1,\vct{p}_t^2)\gamma^1(\vct{p}_t^1 ; \vct{u}_t^1)\gamma^2(\vct{p}_t^2;\vct{u}_t^2).
\end{align}}
{$\tilde{c}_t(\pi,\gamma^1,\gamma^2)$ is the expected value of the  cost at time $t$ if the state and private informations have $\pi$ as their probability distribution and $\gamma^1,\gamma^2$ are the prescriptions chosen by the virtual players.}
\end{definition}
\begin{lemma}\label{infstate}
For any strategy profile $(\tilde{\vct{\chi}}^{1}, \tilde{\vct{\chi}}^{2})$, the common information based belief $\Pi_t$ evolves almost surely as
\begin{align}
\Pi_{t+1} = F_t(\Pi_t, \Gamma_t^{1:2},\vct{Z}_{t+1}), ~~ t \geq 1, \label{eq:pieq}
\end{align}
where $F_t$ is a fixed transformation that does not depend on the virtual players' strategies. Further, the total expected cost can be expressed as
\begin{align}
{\mathcal{J}}(\tilde{\vct{\chi}}^1,\tilde{\vct{\chi}}^2)=\E^{(\tilde{\vct{\chi}}^1,\tilde{\vct{\chi}}^2)}\left[\sum_{t=1}^T \tilde{c}_t(\Pi_t,\Gamma_t^1,\Gamma_t^2)\right],
\end{align}
where $\tilde{c}_t$ is as defined in equation (\ref{tildec}).
\end{lemma}
\begin{proof}
See Appendix \ref{infstateproof}.\qed
\end{proof}
\begin{remark}\label{fremark}
{Because \eqref{eq:pieq} is an almost sure equality, the transformation $F_t$ in Lemma \ref{infstate} is not necessarily unique. In Appendix \ref{infstateproof}, we identify a class of transformations such that for any transformation $F_t$ in this class, Lemma \ref{infstate} holds. We denote this class by $\mathscr{B}$.}
\end{remark}

We now describe two dynamic programs, one for each virtual player in $\gm{G}_e$.
\subsubsection{The min-max dynamic program}
The minimizing virtual player (virtual player 1) in game $\gm{G}_e$ solves the following dynamic program. Define $V^u_{T+1}(\pi_{T+1}) = 0$ for every  $\pi_{T+1}$. In a backward inductive manner, at each time $t \leq T$ and for each possible  common information belief $\pi_t$ and prescriptions $\gamma^1_t, \gamma^2_t$, define 
\begin{align}
w^u_t(\pi_t,\gamma_t^1,\gamma_t^2) &:= \tilde{c}_t(\pi_t,\gamma_t^1,\gamma_t^2) + \E[V^u_{t+1}(F_t(\pi_{t},\gamma_t^{1:2},\rv{Z}_{t+1}))\mid \pi_t,\gamma_t^{1:2}]\\
\label{minequa}V_t^u(\pi_t) &:= \inf_{{\gamma}_t^1} \sup_{\gamma_t^2}w_t^u(\pi_t,\gamma_t^1,\gamma_t^2).
\end{align}

\subsubsection{The max-min dynamic program}
The maximizing virtual player (virtual player 2) in game $\gm{G}_e$ solves the following dynamic program. Define $V^l_{T+1}(\pi_{T+1}) = 0$ for every  $\pi_{T+1}$. In a backward inductive manner, at each time $t \leq T$ and for each possible  common information belief $\pi_t$ and prescriptions $\gamma^1_t, \gamma^2_t$, define 
\begin{align}
w^l_t(\pi_t,\gamma_t^1,\gamma_t^2) &:= \tilde{c}_t(\pi_t,\gamma_t^1,\gamma_t^2)  + \E[V^l_{t+1}(F_t(\pi_{t},\gamma_t^{1:2},\rv{Z}_{t+1}))\mid \pi_t,\gamma_t^{1:2}]\\
\label{maxequa}V_t^l(\pi_t) &:= \sup_{\gamma_t^2}\inf_{{\gamma}_t^1}w_t^l(\pi_t,\gamma_t^1,\gamma_t^2).
\end{align}


\begin{lemma}\label{equiexistlemma}
{For any realization of common information based belief $\pi_t$, the \emph{inf} and \emph{sup} in \eqref{minequa} are achieved, i.e. there exists a measurable mapping $\Xi^1_t: \mathcal{S}_t \rightarrow \mathcal{B}_t^1$ such that
\begin{align}
V_t^u(\pi_t) &=  \min_{\gamma_t^1}\max_{\gamma_t^2}w_t^u(\pi_t,\gamma_t^1,\gamma_t^2)= \max_{\gamma_t^2}w_t^u(\pi_t,\Xi_t^1(\pi_t),\gamma_t^2).
\end{align}
Similarly, for any realization of common information based belief $\pi_t$, the \emph{sup} and \emph{inf} in \eqref{maxequa} are achieved, i.e,  there exists a measurable mapping $\Xi_t^2: \mathcal{S}_t \rightarrow \mathcal{B}_t^2$ such that
\begin{align}
V_t^l(\pi_t) &=  \max_{\gamma_t^2}\min_{\gamma_t^1}w_t^l(\pi_t,\gamma_t^1,\gamma_t^2) =\min_{{\gamma}_t^1}w_t^l(\pi_t,\gamma_t^1,\Xi_t^2(\pi_t)).
\end{align}
}
\end{lemma}
\begin{proof}
See Appendix \ref{equiexistlemmaproof}.
\end{proof}

{\begin{definition}\label{stratdef} Define strategies $\tilde{\chi}^{1*}$ and $\tilde{\chi}^{2*}$ for virtual players 1 and 2 respectively as follows: for each instance of common information $c_t$ and prescription history $\gamma_{1:t-1}^{1:2}$, let
\begin{align}
\tilde{\chi}^{1*}_t(c_t,\gamma_{1:t-1}^{1:2}) &:= \Xi_t^1(\pi_t)\\
\tilde{\chi}^{2*}_t(c_t,\gamma_{1:t-1}^{1:2}) &:= \Xi_2^1(\pi_t),
\end{align}
where $\Xi_t^1$ and $\Xi_t^2$ are the mappings defined in Lemma \ref{equiexistlemma} and $\pi_t$ (which is a function of $c_t,\gamma_{1:t-1}^{1:2}$) is obtained in a forward inductive manner using the relation
\begin{align}
\pi_1(x_1,p_1^1,p_1^2) &= \Py[X_1 = x_1,P_1^1 = p_1^1,P_1^2 = p_1^2 \mid C_1 = c_1] ~ \forall\; x_1,p_1^1,p_1^2,\\
\pi_{\tau + 1} &= F_\tau(\pi_\tau, \gamma_\tau^1,\gamma_\tau^2,z_{\tau+1}), ~  1 \leq \tau < t.
\end{align}
Note that $F_\tau$ is the common information belief update function defined in Lemma \ref{infstate}.
\end{definition}}


The following theorem establishes that the two dynamic programs described above characterize the upper and lower values of game $\gm{G}_e$.
\begin{theorem}\label{dp}
The upper and lower values of the expanded virtual game $\gm{G}_e$ are given by 
\begin{align}
S^u(\gm{G}_e) &= \E[V_1^u(\Pi_1)],\\
S^l(\gm{G}_e) &= \E[V_1^l(\Pi_1)].
\end{align}
Further, the strategies $\tilde{\chi}^{1*}$ and $\tilde{\chi}^{2*}$ as defined in Definition \ref{stratdef} are, respectively, min-max and max-min strategies in the expanded virtual game $\gm{G}_e$.
\end{theorem} 
\begin{proof}
See Appendix \ref{dpproof}.\qed
\end{proof}

{Theorem \ref{dp} gives us a dynamic programming characterization of the upper and lower values of the expanded game. As mentioned in Theorem \ref{origvirt}, the upper and lower values of the expanded game provide bounds on the corresponding values of the original game. Further, if the original game has a Nash equilibrium, the dynamic programs of Theorem \ref{dp} characterize the value of the game. Note that this applies to any dynamic game of  the form in Section \ref{sec:probform} where the common information is non-decreasing in time and the private information has a ``state-like'' update equation (see Assumption \ref{infevolve}). As noted before, a variety of information structures satisfy this assumption \cite{nayyar2013decentralized}, \cite{nayyar2014common}.   }

{ The computational burden of solving the dynamic programs of Theorem \ref{dp} would depend on the specific  information structure being considered, i.e., on the exact nature of common and private information. At one extreme, we can consider the following instance of the original game $\gm{G}$: $C_t = (X_{1:t}), P^1_t=P^2_t = \emptyset$. It is easy to see that in this case, the common information belief can be replaced by the current state in the dynamic programs and the prescriptions are simply distributions on the players' finite action sets. Also, in this case,  $w^u_t$ and $w^l_t$ are bilinear functions of the prescriptions and the min-max/max-min problems at each stage of the dynamic program can be solved by a linear program \cite{ponssard1980lp}. On the other extreme, we can  consider an instance of game $\gm{G}$ with $C_t=\emptyset, P^i_t=Y^i_{1:t}, i=1,2$. In this case, the common information belief will be on the current state and observation histories of the two players and the prescriptions will take values in a large-dimensional space. Also, the functions $w_t^u$ and $w_t^l$ (for $t < T$) in this case do no have any apparent structure that can be exploited for efficient computation of the min-max and max-min values in the dynamic program. One general approach that can be used for any instance of game $\gm{G}$ is to discretize the CIB belief space and compute approximate value functions $V_t^u$ and $V_t^l$  in a backward inductive manner. However, we believe that significant structural and computational insights can be obtained by specializing the dynamic programs of Theorem \ref{dp} to the specific instance of the game being considered. We demonstrate this in the next section where   
 we discuss an information structure where one player has complete information while the other player has only partial information. We will show that in this case, the functions $w_t^u$ and $w_t^l$ turn out to be identical at all times $t$ and they satisfy some structural properties that can be leveraged for computation. Further, we will show that for this information structure, the dynamic programming characterization in Theorem \ref{dp} allows us to find an equilibrium strategy for player 1 in the original game $\gm{G}$.  }

{\paragraph{Comparison with \cite{ouyang2017dynamic} and \cite{vasal2019systematic}:} In \cite{ouyang2017dynamic}, the authors considered an $n$-player stochastic game model which can potentially be non-zero sum. In this model, each player has a private state that is privately observed by the corresponding player and a public state that is commonly observed by all the players. The model in \cite{ouyang2017dynamic} additionally allows players' private information to be partially revealed in the form of common observations. The actions of all the players in this model are commonly observed. The authors also make the assumption that the evolution of the private states of the players is conditionally independent. The model in \cite{vasal2019systematic} can be viewed as a special case of the model in \cite{ouyang2017dynamic}. For these models, backward inductive algorithms were presented to compute perfect Bayesian equilibria. Consider the case when the number of players in the games of \cite{ouyang2017dynamic} and \cite{vasal2019systematic} is two and the games are zero-sum. Then:
\begin{enumerate}
\item The models in \cite{ouyang2017dynamic} and \cite{vasal2019systematic} can be viewed as special cases of our model in Section \ref{sec:probform}.
\item The players in these games have perfect recall. Hence, we can use Kuhn's theorem to conclude that a Nash equilibrium and, thus, the value exists for these zero-sum games. Therefore, we can use the dynamic programs in Section \ref{dpsec} to the characterize the value of these zero-sum games. This characterization does not make any additional assumptions. The backward inductive algorithms in \cite{ouyang2017dynamic} and \cite{vasal2019systematic}, however, require the existence of a particular kind of \emph{fixed point} solution at each stage. This fixed point solution is not guaranteed to exist in general. Thus, there may be instances where the approaches in \cite{ouyang2017dynamic} and \cite{vasal2019systematic} fail to characterize the value of the game while our dynamic program in Section \ref{dpsec} can always characterize it.
\end{enumerate} 
}

%% file: oneside.tex
In this section, we consider a special case of the original game $\gm{G}$ (described in Section \ref{sec:probform})  where player 1 has complete information, that is, it  knows  the entire state history as well as the observation and action histories of both players. On the other hand, player 2 has only partial information on the state and player 1's action history. For  games with this information structure, we show that (i) a Nash equilibrium, and hence the value of the game, exist; (ii) the value can be computed using the methodology proposed in Section \ref{sec:expanded}; (iii) the dynamic program in Section \ref{sec:expanded} also characterizes an equilibrium strategy for player 1; (iv) the value functions of the dynamic program satisfy some key structural properties that may be useful for computational purposes.

\subsection{System Model}

Consider a system with state evolution as in equation (\ref{statevol}).  
At each time $t$, player 1 observes the state perfectly  but player 2 gets an imperfect observation $Y^2_t$ defined as follows:
\begin{equation}
\rv{Y}_t^2 = h_t^2(\rv{X}_t, \rv{U}_{t-1}^{1}, \rv{U}_{t-1}^{2},\rv{W}_t^2). \label{eq:sec5eq1}
\end{equation}
Player 1 has complete information: at each time $t$, it knows the entire state, observation and action histories. Player 2 has partial information: at each time $t$, it knows only the observation history $Y^2_{1:t}$ and its own action history $U^2_{1:t-1}$.
Thus, the total information available to each player at $t$ is as follows:
\begin{align}
\rv{I}_t^1 &= \{\rv{X}_{1:t}, \rv{Y}_{1:t}^{2},\rv{U}_{1:t-1}^{1}, \rv{U}_{1:t-1}^{2}\}\\
\rv{I}_t^2 &= \{\rv{Y}_{1:t}^{2},\rv{U}_{1:t-1}^{2}\}.
\end{align}
Clearly, $\rv{I}_t^2 \subseteq \rv{I}_t^1$, that is, player 1 is more informed than player 2. The  common and private information for this game can be written as follows: $\rv{C}_t = \rv{I}_t^2$,  $\rv{P}_t^1 = \{\rv{X}_{1:t},\rv{U}^1_{1:t-1}\}$ and $\rv{P}_t^2 = \varnothing$. The increment in common information at time $t$ is $\rv{Z}_t = \{\rv{Y}_t^2,\rv{U}_{t-1}^2\}$. The total expected cost is as defined in \eqref{eq:totalcost}. As before, Player 1 wants to minimize the total expected cost, while Player 2 wants to maximize it. Players' strategies are as described in \eqref{eq:stg}. {We refer to this game as Game $\gm{G}^{\mathrm{special}}$.}
{\begin{remark}
Consider a scenario in which player 1's information at time $t$ is $\{X_{1:t},U_{1:t-1}^1,U_{1:t-1}^2\}$. Player 2's information at time $t$ is $\{Y_{1:t}^2,U_{1;t-1}^2\}$ and its observation $Y_t^2$ is of the form
$Y_t^2 = h_t^2(X_t,U_{t-1}^1,U_{t-1}^2).$ That is, player 2 observes a \emph{quantized} version of the current system state and player 1's previous action. At first glance, it might appear that the second player's private information for this scenario is $P_t^2 = Y_{1:t}^2$ and, thus, it does not fit the model in Game $\gm{G}^{\mathrm{special}}$. However, player 1 knows the function $h_t^2$ and the variables $X_t,U_{t-1}^1,U_{t-1}^2$. Thus, it can simply compute $Y_{1:t}^2$ and select its behavior action using $\{X_{1:t}, Y_{1:t}^2,U_{1:t-1}^1,U_{1:t-1}^2\}$. In that case, $P_t^2 = \varnothing$ and, therefore, this scenario fits the model in Game $\gm{G}^{\mathrm{special}}$.
\end{remark}}
\begin{remark}
{The model described above subsumes the system dynamics and information structures in \cite{renault2006value,li2014lp,zheng2013decomposition} and \cite{renault2012value}. In all these works (except \cite{renault2012value}), the less-informed player does not have any state information and can only observe the more-informed player's action. That is, in these works, $Y_t^2 = U_{t-1}^1$ (which is a special case of \eqref{eq:sec5eq1} above) and thus, $\rv{I}_t^2 = \{\rv{U}_{1:t-1}^{1},\rv{U}_{1:t-1}^{2}\}$}. Further, the system dynamics in some of these models are more specialized. In particular, in \cite{renault2006value}, the system is uncontrolled and in \cite{li2014lp} and \cite{renault2012value}, the less-informed player cannot control the state. Note that our model allows for both players to influence the system dynamics (see \eqref{statevol}).
\end{remark}
The following lemmas establish the existence   and a structural property of Nash equilibrium strategies in game $\gm{G}^{\text{special}}$.
\begin{lemma}
A Nash equilibrium exists in game $\gm{G}^{\mathrm{special}}$.
\end{lemma}
\begin{proof}
Notice that the players have perfect recall. Therefore, as discussed in Remark \ref{kuhnremark}, the existence of a behavioral strategy Nash equilibrium is guaranteed.\qed
\end{proof}
\begin{lemma}\label{structlemma}
There exists a Nash equilibrium $(g^{1*},g^{2*})$ such that the control law $g^{1*}_t$ selects player 1's behavioral action using only $\rv{X}_t$ and  $\rv{I}_t^2$, i.e,
\[\delta U^1_t = g^{1*}_t(X_t,I^2_t).\]
\end{lemma}
\begin{proof}
See Appendix \ref{structlemmaproof}.\qed
\end{proof}
The lemma above implies that, for the purpose of characterizing the value of the game and a min-max strategy for player 1, we can restrict player 1's information structure to be $I^1_t = \{X_t,I^2_t\}$.
Thus, the common and private information become:  $\rv{C}_t = \rv{I}_t^2$, $\rv{P}_t^1 = \{\rv{X}_{t}\}$ and $\rv{P}_t^2 = \varnothing$. 
We refer to this game with reduced private information as Game $\gm{H}$. The corresponding virtual game and expanded virtual game are denoted by $\gm{H}_v$ and $\gm{H}_e$ respectively.
\subsection{Expanded Game and the Dynamic Program}\label{dponesided}

We can now proceed with the construction of the expanded virtual game $\gm{H}_e$ as outlined in Section \ref{sec:expanded}. Recall that in the expanded virtual game, virtual players select prescriptions based on common information and the prescription history. 

{\subsubsection{Belief update}\label{beliefoneside}
{In Game $\gm{H}$, player 2 has no private information and player 1's private information $P_t^1$ at time $t$ is just the state $X_t$. Therefore, in the expanded virtual game $\gm{H}_e$, the common information belief on $(X_t,P_t^1,P_t^2)$ is equivalent to the common information belief on $X_t$ alone. Thus, in game $\gm{H}_e$, the common information belief $\Pi_t$ is a vector of size $|\mathcal{X}_t|$.} As noted in Remark \ref{fremark}, the common information belief update rule $F_t$ is not unique and there is an entire class $\mathscr{B}$ of update rules for which Lemma \ref{infstate} holds. When this class $\mathscr{B}$ is specialized to the Game $\gm{H}_e$, there exists an update rule $F_t$ in this class that does not use virtual player 2's prescription $\gamma_t^2$. The following lemma makes this precise.
\begin{lemma}\label{beliefupdateone}
For any strategy profile $(\tilde{\vct{\chi}}^{1}, \tilde{\vct{\chi}}^{2})$, the common information belief $\Pi_t$ evolves almost surely as
\begin{equation}
\Pi_{t+1} = F_t(\Pi_t, \Gamma_t^1, Z_{t+1}), t \geq 1,
\end{equation}
where $F_t$ is a fixed transformation that does not depend on the virtual players' strategies.
\end{lemma}
\begin{proof}
See Appendix \ref{beliefupdateoneproof}.\qed
\end{proof}}
\begin{remark}
{In \cite{renault2006value}, \cite{renault2012value} and \cite{li2014lp}, it was assumed that {the less informed player cannot} control the state evolution. In that case, Lemma \ref{beliefupdateone} is easy to establish. We would like to emphasize that Lemma \ref{beliefupdateone} is true for our more general model where both players control the state evolution.}
\end{remark}

\subsubsection{Dynamic program}\label{dp:oneside}
{Since a Nash equilibrium exists in game $\gm{H}$, the upper and lower values of the corresponding expanded virtual game $\gm{H}_e$ are both equal to the value of game $\gm{H}$ (see Theorem \ref{origvirt}). According to Theorem \ref{dp}, this value can be obtained by solving either the min-max dynamic program or the max-min dynamic program as discussed in Section \ref{dpsec}. The dynamic programs in Section \ref{dpsec}, when specialized to the game $\gm{H}_e$, are as follows.}

{\paragraph{The min-max dynamic program:}
The minimizing virtual player (virtual player 1) in game $\gm{H}_e$ solves the following dynamic program. Define $V^u_{T+1}(\pi_{T+1}) = 0$ for every  $\pi_{T+1}$. In a backward inductive manner, at each time $t \leq T$ and for each possible  common information belief $\pi_t$ and prescriptions $\gamma^1_t, \gamma^2_t$, define 
\begin{align}
w^u_t(\pi_t,\gamma_t^1,\gamma_t^2) &:= \tilde{c}_t(\pi_t,\gamma_t^1,\gamma_t^2) + \E[V^u_{t+1}(F_t(\pi_{t},\gamma_t^{1},\rv{Z}_{t+1}))\mid \pi_t,\gamma_t^{1:2}]\\
\label{minequa:one}V_t^u(\pi_t) &:= \inf_{{\gamma}_t^1} \sup_{\gamma_t^2}w_t^u(\pi_t,\gamma_t^1,\gamma_t^2).
\end{align}
{Let $\Xi_t^1(\pi_t)$ be the min-maximizer in \eqref{minequa:one}. Note that the existence of a min-maximizer and a measurable mapping $\Xi_t^1$ is guaranteed by Lemma \ref{equiexistlemma}.}}

{\paragraph{The max-min dynamic program:}
The maximizing virtual player (virtual player 2) in game $\gm{H}_e$ solves the following dynamic program. Define $V^l_{T+1}(\pi_{T+1}) = 0$ for every  $\pi_{T+1}$. In a backward inductive manner, at each time $t \leq T$ and for each possible  common information belief $\pi_t$ and prescriptions $\gamma^1_t, \gamma^2_t$, define 
\begin{align}
w^l_t(\pi_t,\gamma_t^1,\gamma_t^2) &:= \tilde{c}_t(\pi_t,\gamma_t^1,\gamma_t^2)  + \E[V^l_{t+1}(F_t(\pi_{t},\gamma_t^{1},\rv{Z}_{t+1}))\mid \pi_t,\gamma_t^{1:2}]\\
\label{maxequa:one}V_t^l(\pi_t) &:= \sup_{\gamma_t^2}\inf_{{\gamma}_t^1}w_t^l(\pi_t,\gamma_t^1,\gamma_t^2).
\end{align}
{Let $\Xi_t^2(\pi_t)$ be the max-minimizer in \eqref{maxequa:one}. Note that the existence of a max-minimizer and a measurable mapping $\Xi_t^2$ is guaranteed by Lemma \ref{equiexistlemma}.}}

{\begin{definition}\label{stratdefone} Define strategies $\tilde{\chi}^{1*}$ and $\tilde{\chi}^{2*}$ for virtual players 1 and 2, respectively, as follows: for each instance of common information $c_t$ and prescription history $\gamma_{1:t-1}^{1:2}$, let
\begin{align}
\tilde{\chi}^{1*}_t(c_t,\gamma_{1:t-1}^{1:2}) &:= \Xi_t^1(\pi_t)\\
\tilde{\chi}^{2*}_t(c_t,\gamma_{1:t-1}^{1:2}) &:= \Xi_2^1(\pi_t),
\end{align}
where $\Xi_t^1$ and $\Xi_t^2$ are respectively the mappings obtained from the min-max and the max-min dynamic programs stated above and, $\pi_t$ (which is a function of $c_t,\gamma_{1:t-1}^1$) is obtained in a forward inductive manner using the relation
\begin{align}
\pi_1(x_1) &= \Py[X_1 = x_1 \mid C_1 = c_1] ~ \forall \; x_1,\\
\pi_{\tau + 1} &= F_\tau(\pi_\tau, \gamma_\tau^1,z_{\tau+1}), ~ 1 \leq \tau < t.
\end{align}
Note that $F_\tau$ is the common information belief update function defined in Lemma \ref{beliefupdateone}.
\end{definition}}

{\begin{lemma}\label{valthmone}
The value of the game $\gm{H}$ is given by
\begin{align}
S(\gm{H}) = \E[V_1^u[\Pi_1]] = \E[V_1^l[\Pi_1]],
\end{align}
where $V_1^u$ and $V_1^l$ are respectively the value functions at time $t=1$ in the min-max and the max-min dynamic program described above.{ Further, the strategies $\tilde{\chi}^{1*}$ and $\tilde{\chi}^{2*}$ as defined in Definition \ref{stratdefone} are, respectively, min-max and max-min strategies in the expanded virtual game $\gm{H}_e$.}
\end{lemma}
\begin{proof}
This is a direct consequence of Theorem \ref{dp}. \qed
\end{proof}}

{Thus, the dynamic program for Game $\gm{H}$ can be obtained by simply specializing the approach used for the more general Game $\gm{G}$. In the subsequent sections, we will derive some additional properties of the dynamic program in Section \ref{dp:oneside} that are specific to the model of Game $\gm{H}$.}

\subsection{An Equilibrium Strategy for the More-informed Player}



The following result shows that the mapping $\Xi_t^1$ obtained by solving the min-max dynamic program in Section \ref{dp:oneside} can be used to construct an equilibrium strategy for Player 1 in the original game $\gm{H}$.
{\begin{theorem}\label{strategy}
Define $\vct{g}^{1*}$ for player 1 in game $\gm{H}$ such that for every $t \leq T$ and $\vct{p}_t^1,\vct{c}_t$, 
\begin{align}
g_t^{1*}(\vct{p}_t^{1},\vct{c}_t) = [\Xi_t^1(\pi_t)](p_t^1),
\end{align}
where $\Xi_t^1$ is as defined in the min-max dynamic program in Section \ref{dp:oneside} and, $\pi_t$ is computed in a forward inductive manner using the following relation
\begin{align}
\pi_1(x_1) &= \Py[X_1 = x_1 \mid C_1 = c_1] ~ \forall \; x_1\\
\pi_{\tau + 1} &= F_\tau(\pi_\tau, \Xi_\tau^1(\pi_\tau),z_{\tau+1}), ~ 1 \leq \tau < t.
\end{align}
Then $\vct{g}^{1*}$ is an equilibrium (minimax) strategy for player 1 in the original game $\gm{H}$.
\end{theorem}}
\begin{proof}
See Appendix \ref{strategyproof}.\qed
\end{proof}
{Theorem \ref{strategy} proves the existence of a Nash equilibrium $(g^{1*},g^{2*})$ in the game $\gm{H}$, where player 1's behavioral action at time $t$ is a function of its private information $P_t^1 = X_t$ and the common information belief $\Pi_t$.}

\subsection{Computational Aspects}
The value functions in the dynamic programs in Section \ref{dp:oneside} admit the following structural property.
\begin{lemma}[Piecewise linearity and convexity]\label{piecelemma}
At each time $t$, 
\begin{enumerate}
\item The min-max value function $V_t^u$ and the max-min value function $V_t^l$ are identical. Consequently, the functions $w_t^u$ and $w_t^l$ are also identical. Let $V_t := V_t^l = V_t^u$ and $w_t := w_t^u = w_t^l$ for each $t$.
\item There exists a \emph{finite} collection $\mathcal{A}_t$ of vectors of size $|\mathcal{X}_t|$ such that
\begin{align}
V_t(\pi_t) = \max_{\ell \in \mathcal{A}_t}\langle \ell, \pi_t \rangle.
\end{align}
Consequently, the value functions are piecewise linear and convex functions of the belief $\pi_t$.
\item {The function $w_t(\pi_t,\gamma_t^1,\gamma_t^2)$ is linear in $\gamma_t^2$ for any given $\pi_t,\gamma_t^1$; and it is convex in $\gamma_t^1$ for any given $\pi_t,\gamma_t^2$.}
\end{enumerate}
\end{lemma}
\begin{proof}
See Appendix \ref{piecelemmaproof}.\qed
\end{proof}
The above structural property of value functions allows us to approximately solve the min-max or the max-min dynamic program using a method described below. This method is inspired by the function approximation based methods for dynamic programming in \cite{bertsekas1996neuro}.

At each time $t$, let $\hat{V}_t(\pi_t,\theta_t)$ be a representation of the value function $V_t(\pi_t)$ in parametric form, where $\theta_t$ is a vector representing the parameters. Note that this representation is only an approximation of the true value function $V_t(\pi_t)$ and may not be exactly equal to it. Since we know that the value function $V_t$ is piecewise linear and convex, we can enforce the parametric representation $\hat{V}_t$ to be piecewise linear and convex in $\pi_t$ for any $\theta_t$. One such representation is as follows:
\begin{align}\label{valrep}
\hat{V}_t(\pi_t,\theta_t) = \max_{\ell_t \in \hat{\mathcal{A}_t}}\langle \ell_t,\pi_t \rangle
\end{align}
where $\hat{\mathcal{A}_t}$ is a finite collection of vectors of size $|\X_t|$. Here, $\theta_t$ is the column vector obtained by stacking all the vectors $\ell \in \hat{\mathcal{A}_t}$.

 The parameters $\theta_t$ are chosen in a backward inductive manner such that the function $\hat{V}_t^l(\cdot, \theta_t)$ is close to the true value function $V_t$. More precisely, at time $t$, we uniformly (either randomly or deterministically) sample belief vectors from the simplex $\mathcal{S}_t$ to obtain a finite collection of belief vectors $\pi_t$. Let this collection be $\mathscr{S}_t$. For each belief vector $\pi_t \in \mathscr{S}_t$, we will compute $\bar{V}_t(\pi_t)$ which is given by
{\begin{align} 
\hat{w}_t(\pi_t,\gamma_t^1,\gamma_t^2) &:= \tilde{c}_t(\pi_t,\gamma_t^1,\gamma_t^2) +\E[\hat{V}_{t+1}(F_t(\pi_{t},\gamma_t^{1},\rv{Z}_{t+1}),\theta_{t+1})\mid \pi_t,\gamma_t^{1:2}]\\
\bar{V}_t(\pi_t) &:= \min_{{\gamma}_t^1}\max_{\gamma_t^2}\hat{w}_t(\pi_t,\gamma_t^1,\gamma_t^2).\label{barv}
\end{align}}
{Since $\hat{V}_{t+1}$ is piecewise linear and convex in $\pi_{t+1}$ by construction, we can use the same arguments as in the proof (see Appendix \ref{pieceproof}) of Lemma \ref{piecelemma} to conclude that $\hat{w}_t(\pi_t,\gamma_t^1,\gamma_t^2)$ is linear in $\gamma_t^2$ and convex in $\gamma_t^1$. Exploiting these properties, the min-max problem in \eqref{barv} can be formulated as
\begin{equation*}
\begin{aligned}
& \underset{\gamma_t^1 \in \mathcal{B}_t^1, \nu \in \R}{\text{min}}
 & & \nu \\
& \text{subject to }
 & &\hat{w}_t(\pi_t,\gamma_t^1, \mathbbm{1}(u_t^2))  \leq \nu, \; \forall u_t^2 \in \mathcal{U}_t^2,
\end{aligned}
\end{equation*}
where $\mathbbm{1}(u_t^2)$ denotes the prescription with $\gamma_t^2(u_t^2) = 1$. The optimization problem stated above is a convex program and, since we have a closed-form representation of the functions $\hat{w}_t$, this convex program can be solved using standard convex optimization algorithms \cite{ben2001lectures,boyd2004convex}.} 

We can then solve the following regression (least-squares) problem to obtain a piecewise linear and convex representation of $\hat{V}_t$
\begin{align}\label{regression}
\min_{\theta_t}\sum_{\pi_t \in \mathscr{S}_t} (\hat{V}(\pi_t,\theta_t) - \bar{V}_t(\pi_t))^2.
\end{align}
Thus, at each time $t$, we have an approximate representation of the value function $V_t$.

{\begin{remark}
Our methodology for solving the dynamic program is similar to the neural network based approximate dynamic programming methods in \cite{bertsekas1996neuro}. In this neural network based approach, an approximate version of the value function is represented as a neural network. A finite collection of \emph{samples} of the true value function is obtained and the neural network is then \emph{trained} on these samples, that is, the weights of the neural network are adjusted such that the parametric representation is close to the true value function. This training is achieved via the regression problem described in \eqref{regression}. Note that by designing the architecture appropriately, it is possible to ensure that the output of the neural network is a piecewise linear and convex function of the input. In fact, our representation in \eqref{valrep} can be viewed as a single layer of neurons with linear activation followed by a max-pooling layer \cite{schmidhuber2015deep}.
\end{remark}}

In general, there may be a trade-off between the computational burden and approximation error associated with this method. That is, to achieve lower approximation error, we may have to sample a very large number of beliefs. This in turn increases the computational complexity of performing the regression in the approach described above. Understanding the precise relationship between approximation error and computational complexity is a problem for future work and beyond the scope of this paper.

%% file: conc.tex
In this paper, we considered a general model of zero-sum stochastic games with asymetric information. For this general model, we provided a dynamic programming approach for
characterizing the value (if it exists). This dynamic programming characterization of value relies on our construction of two {virtual games} that have the same value as our original game. If the value does not exist in the original game, then our dynamic program provides bounds on the upper and lower values of the original game. We then focused on game models in which one player has complete game information and the other has partial information. For such games, we showed that the value exists and used our dynamic programming approach to characterize the value. Further, we proved the existence of a Nash equilibrium in such games where the more-informed player plays a common information belief based strategy. We discussed a dynamic programming approach for computing a common information belief based equilibrium strategy for the more-informed player.

%% file: append_cdc.tex

\section{Proof of Lemma \ref{virtlemma}}\label{virtlemmaproof}
It was shown in \cite{nayyar2014common} that there exist \emph{bijective} mappings $\mathcal{M}^i: \mathcal{G}^i \rightarrow \mathcal{H}^i$, $i =1,2,$ such that for every $g^1 \in \G^1$ and $g^2 \in \G^2$, we have
\begin{align}
J(g^1,g^2) = \mathcal{J}(\mathcal{M}^1(g^1),\mathcal{M}^2(g^2)).
\end{align}
Therefore, for any strategy $g^1 \in \G^1$, we have
\begin{align}
\sup_{g^2\in \G^2}J(g^1,g^2) &= \sup_{g^2\in \G^2}\mathcal{J}(\mathcal{M}^1(g^1),\mathcal{M}^2(g^2))\\
&= \sup_{\chi^2\in \H^2}\mathcal{J}(\mathcal{M}^1(g^1),\chi^2).
\end{align}
Consequently,
\begin{align}
\inf_{g^1\in\G^1}\sup_{g^2\in \G^2}J(g^1,g^2) &= \inf_{g^1\in\G^1}\sup_{\chi^2\in \H^2}\mathcal{J}(\mathcal{M}^1(g^1),\chi^2)\\
&=  \inf_{\chi^1\in\H^1}\sup_{\chi^2\in \H^2}\mathcal{J}(\chi^1,\chi^2).
\end{align}
This implies that $S^u(\gm{G}) = S^u(\gm{G}_v)$. We can similarly prove that $S^l(\gm{G}) = S^l(\gm{G}_v)$.
{\begin{remark}
We can also show that a strategy profile $(g^1,g^2)$ is a Nash equilibrium in game $\gm{G}$ if and only if $(\mathcal{M}^1(g^1),\mathcal{M}^2(g^2))$ is a Nash equilibrium in game $\gm{G}_v$.
\end{remark}}

\section{Proof of Lemma \ref{evolequi}}
\label{evolequiproof}

Let us consider the evolution of the virtual game $\gm{G}_v$ under the strategy profile $(\chi^1,\chi^2)$ and the expanded virtual game $\gm{G}_e$ under the strategy profile $(\tilde{\chi}^1,\tilde{\chi}^2)$. Let the primitive variables and the randomization variables $K_t^i$ in both games be identical. The variables such as the state, action and information variables in the expanded game $\gm{G}_e$ are distinguished from those in the virtual game $\gm{G}_v$ by means of a tilde. For instance, $X_t$ is the state in game $\gm{G}_v$ and $\tilde{X}_t$ is the state in game $\gm{G}_e$.

We will prove by induction that the system evolution in both these games is identical over the entire horizon. This is clearly true at the end of time $t=1$ because the state, observations and the common and private information variables are identical in both games. Moreover, since $\vct{\chi}^i = \varrho^i(\vct{\tilde{\chi}}^1,\vct{\tilde{\chi}}^2)$, $i=1,2$, the strategies $\chi^i_1$ and $\tilde{\chi}^i_1$ are identical by definition (see Definition \ref{def:rho}). Thus, the prescriptions and actions at $t=1$ are also identical.

For induction, assume that the system evolution in both games is identical until the end of time $t$. Then, $$\rv{X}_{t+1} = f_t(\rv{X}_t, \rv{U}_t^{1:2},\rv{W}_t^s) = f_t(\tilde{\rv{X}}_t, \tilde{\rv{U}}_t^{1:2},\rv{W}_t^s) = \tilde{\rv{X}}_{t+1}.$$ Using equations (\ref{obseq}), (\ref{privevol}) and (\ref{commonevol}), we can similarly argue that $\rv{Y}_{t+1}^i = \tilde{\rv{Y}}_{t+1}^i$, $\rv{P}_{t+1}^i = \tilde{\rv{P}}_{t+1}^i$ and $\rv{C}_{t+1} = \tilde{\rv{C}}_{t+1}$. Since $\vct{\chi}^i = \varrho^i(\vct{\tilde{\chi}}^1,\vct{\tilde{\chi}}^2)$, we also have 
\begin{align}
\tilde{\Gamma}_{t+1}^i &= \tilde{\chi}_{t+1}^i(\tilde{\rv{C}}_{t+1},\tilde{\Gamma}_{1:t}^{1:2}) \stackrel{a}{=} \chi_{t+1}^i(\tilde{\rv{C}}_{t+1}) \stackrel{b}{=} \Gamma_{t+1}^i.
\end{align}
Here, equality $(a)$ follows from the construction of the mapping $\varrho^i$ (see Definition \ref{def:rho}) and equality $(b)$ follows from the fact that $C_{t+1} = \tilde{C}_{t+1}$. Further, 
\begin{align}
\rv{U}_{t+1}^i = \kappa(\Gamma_{t+1}^i(\rv{P}_{t+1}^i),\rv{K}_{t+1}^i) &= \kappa(\tilde{\Gamma}_{t+1}^i(\tilde{\rv{P}}_{t+1}^i),\rv{K}_{t+1}^i) \\
&= \tilde{\rv{U}}_{t+1}^i.
\end{align}
Thus, by induction, the hypothesis is true for every $1\leq t \leq T$. This proves that the virtual and expanded games have identical dynamics under strategy profiles $(\vct{\chi}^1,\chi^2)$ and $(\tilde{\vct{\chi}}^1,\tilde{\chi}^2)$.

Since the virtual and expanded games have the same cost structure, having identical dynamics ensures that strategy profiles $(\vct{\chi}^1,\chi^2)$ and $(\tilde{\vct{\chi}}^1,\tilde{\chi}^2)$ have the same expected cost in games $\gm{G}_v$ and $\gm{G}_e$, respectively. Therefore, $\mathcal{J}(\vct{\chi}^1,\chi^2) = {\mathcal{J}}(\tilde{\vct{\chi}}^1,\tilde{\chi}^2)$.

\section{Proof of Theorem \ref{origvirt}}\label{origvirtproof}

For any strategy $\chi^1 \in \H^1$, we have 
\begin{align}
\sup_{\tilde{\chi}^2 \in \tilde{\H}^2}\mathcal{J}({\chi}^1,\tilde{\chi}^2) \geq \sup_{{\chi}^2 \in {\H}^2}\mathcal{J}({\chi}^1,{\chi}^2), \label{eq:thm1c}
\end{align}
because $\H^2 \subseteq \tilde{\H}^2$. Further,
\begin{align}
\sup_{\tilde{\chi}^2 \in \tilde{\H}^2}\mathcal{J}({\chi}^1,\tilde{\chi}^2) &=\sup_{\tilde{\chi}^2 \in \tilde{\H}^2}\mathcal{J}(\varrho^1(\chi^1,\tilde{\chi}^2),\varrho^2(\chi^1,\tilde{\chi}^2)). \label{eq:thm1a} \\
&= \sup_{\tilde{\chi}^2 \in \tilde{\H}^2}\mathcal{J}({\chi}^1,\varrho^2(\chi^1,\tilde{\chi}^2))\\
&\leq \sup_{{\chi}^2 \in {\H}^2}\mathcal{J}({\chi}^1,{\chi}^2), \label{eq:thm1b}
\end{align} 
where the first equality is due to Lemma \ref{evolequi}, the second equality is because $\varrho^1(\chi^1,\tilde{\chi}^2) = \chi^1$ and the last inequality is due to the fact that $\varrho^2(\chi^1,\tilde{\chi}^2) \in {\H}^2$ for any $\tilde{\chi}^2 \in \tilde{\H}^2$.

Combining \eqref{eq:thm1c} and \eqref{eq:thm1b}, we obtain that
\begin{align}
\sup_{{\chi}^2 \in {\H}^2}\mathcal{J}({\chi}^1,{\chi}^2) = \sup_{\tilde{\chi}^2 \in \tilde{\H}^2}\mathcal{J}({\chi}^1,\tilde{\chi}^2). \label{eq:thm1d}
\end{align}
Now, 
\begin{align}
S^u(\gm{G}_e) &:=\inf_{\tilde{\chi}^1 \in \tilde{\H}^1}\sup_{\tilde{\chi}^2 \in \tilde{\H}^2}\mathcal{J}(\tilde{\chi}^1,\tilde{\chi}^2) \\
\label{infsupineq} &\leq \inf_{{\chi}^1 \in {\H}^1}\sup_{\tilde{\chi}^2 \in \tilde{\H}^2}\mathcal{J}({\chi}^1,\tilde{\chi}^2)\\
&= \inf_{{\chi}^1 \in {\H}^1}\sup_{{\chi}^2 \in {\H}^2}\mathcal{J}({\chi}^1,{\chi}^2), \label{eq:thm1e}\\
&=: S^u(\gm{G}_v).
\end{align}
where the inequality (\ref{infsupineq}) is true since $\H^1 \subseteq \tilde{\H}^1$ and the   equality  in \eqref{eq:thm1e}  follows from \eqref{eq:thm1d}. 
Therefore, $S^u(\gm{G}_e) \leq S^u(\gm{G}_v)$. We can use similar arguments to show that $S^l(\gm{G}_v) \leq S^l(\gm{G}_e).$

\section{Proof of Lemma \ref{infstate}}\label{infstateproof}
We begin with defining the following transformations for each time $t$. Recall that $\mathcal{S}_t$ is the set of all possible common information beliefs at time $t$ and $\mathcal{B}_t^i$ is the prescription space for virtual player $i$  at time $t$.
\begin{definition}\label{fdef}
\begin{enumerate}[(i)]
\item Let $P_t^j: \mathcal{S}_t \times \mathcal{B}_t^1 \times \mathcal{B}_t^2 \to \Delta(\mathcal{Z}_{t+1} \times \X_{t+1} \times \P_{t+1}^1 \times \P_{t+1}^2)$ be defined as
\begin{align}
P_t^j(\pi_t,&\gamma_t^{1:2}; z_{t+1}, x_{t+1},p_{t+1}^{1:2}) \label{jointprob}\\
&:= \sum_{\vct{x}_t,\vct{p}^{1:2}_t,\vct{u}_t^{1:2}}\pi_t(\vct{x}_t,\vct{p}_t^{1:2})\gamma_t^1(\vct{p}_t^1;u_t^1)\gamma_t^2( \vct{p}_t^2; u_t^2)\Py[\vct{x}_{t+1}, \vct{p}_{t+1}^{1:2},\vct{z}_{t+1} \mid \vct{x}_t,\vct{p}_t^{1:2},\vct{u}_t^{1:2}].\label{condindepeq}
\end{align}
We will use $P_t^j(\pi_t,\gamma_t^{1:2})$ as a shorthand for the probability distribution $P_t^j(\pi_t,\gamma_t^{1:2}; \cdot, \cdot, \cdot)$.
The distribution $P_t^j(\pi_t,\gamma_t^{1:2})$ can be viewed as a joint distribution over the variables $Z_{t+1},X_{t+1}, P_{t+1}^{1:2}$ if the distribution on $X_t, P^{1:2}_t$ is $\pi_t$ and prescriptions $\gamma^{1:2}_t$ are chosen by the virtual players at time $t$.
\item Let $P_t^m: \mathcal{S}_t \times \mathcal{B}_t^1 \times \mathcal{B}_t^2 \to \Delta\mathcal{Z}_{t+1}$ be defined as
\begin{align}
P_t^m(\pi_t,\gamma_t^{1:2}; z_{t+1}) = \sum_{x_{t+1},p_{t+1}^{1:2}} P_t^j(\pi_t,\gamma_t^{1:2}; z_{t+1}, x_{t+1},p_{t+1}^{1:2}). \label{marginalprob}
\end{align}
{The distribution $P_t^m(\pi_t,\gamma_t^{1:2})$ is the marginal distribution of the variable $Z_{t+1}$ obtained from the joint distribution $P_t^j(\pi_t,\gamma_t^{1:2})$ defined above.}
\item Let $F_t: \mathcal{S}_t \times \mathcal{B}_t^1 \times \mathcal{B}_t^2 \times \mathcal{Z}_{t+1} \to \mathcal{S}_{t+1}$ be defined as
\begin{align}
F_t(\pi_t,\gamma_t^{1:2},z_{t+1})= 
\begin{cases}
\frac{P_t^j(\pi_t,\gamma_t^{1:2};z_{t+1},\cdot,\cdot)}{P_t^m(\pi_t,\gamma_t^{1:2};z_{t+1})} &\text{if } P_t^m(\pi_t,\gamma_t^{1:2};z_{t+1}) > 0\\
G_t(\pi_t,\gamma_t^{1:2},z_{t+1}) & \text{otherwise},
\end{cases}
\end{align}
where $G_t: \mathcal{S}_t \times \mathcal{B}_t^1 \times \mathcal{B}_t^2 \times \mathcal{Z}_{t+1} \to \mathcal{S}_{t+1}$ can be any arbitrary measurable mapping. {One such mapping is the one that maps every element $\pi_t,\gamma_t^{1:2},z_{t+1}$ to the uniform distribution over the finite space $\mathcal{X}_{t+1} \times \P_{t+1}^1 \times \P_{t+1}^2$.}
\end{enumerate}
\end{definition}
Let the collection of transformations $F_t$ that can be constructed using the method described in Definition \ref{fdef} be denoted by $\mathscr{B}$. Note that the transformations $P_t^j, P_t^m$ and $F_t$ do not depend on the strategy profile $(\tilde{\chi}^1, \tilde{\chi}^2)$ because the term $\Py[\vct{x}_{t+1}, \vct{p}_{t+1}^{1:2},\vct{z}_{t+1} \mid \vct{x}_t,\vct{p}_t^{1:2},\vct{u}_t^{1:2}]$ in (\ref{condindepeq}) depends only on the system dynamics (see equations (\ref{virdyn1}) -- (\ref{virdyn5})) and not on the strategy profile $(\tilde{\chi}^1,\tilde{\chi}^2)$.

{Consider a strategy profile $(\tilde{\chi}^1, \tilde{\chi}^2)$.  Note that the number of possible realizations of common information and prescription history  under $(\tilde{\chi}^1, \tilde{\chi}^2)$  is finite. } Let $c_{t+1},\gamma_{1:t}^{1:2}$ be a realization of the common information and prescription history at time $t+1$ with non-zero probability of occurrence under  $(\tilde{\chi}^1, \tilde{\chi}^2)$. For this realization of virtual players' information, the common information based belief on the state and private information at time $t+1$ is given by
\begin{align}
\nonumber&\pi_{t+1}(x_{t+1},p_{t+1}^{1:2}) \\
&= \nonumber\Py^{(\tilde{\chi}^1,\tilde{\chi}^2)}[X_{t+1} = \vct{x}_{t+1}, P_{t+1}^{1:2} = \vct{p}_{t+1}^{1:2} \mid \vct{c}_{t+1},\gamma_{1:t}^{1:2}]\\
&= \nonumber\Py^{(\tilde{\chi}^1,\tilde{\chi}^2)}[X_{t+1} = \vct{x}_{t+1}, P_{t+1}^{1:2} = \vct{p}_{t+1}^{1:2}\mid \vct{c}_{t},\gamma_{1:t-1}^{1:2},\vct{z}_{t+1},\gamma_t^{1:2}]\\
\label{b1}&= \frac{\Py^{(\tilde{\chi}^1,\tilde{\chi}^2)}[X_{t+1} = \vct{x}_{t+1}, P_{t+1}^{1:2} = \vct{p}_{t+1}^{1:2},Z_{t+1} = \vct{z}_{t+1} \mid \vct{c}_{t},\gamma_{1:t}^{1:2}]}{\Py^{(\tilde{\chi}^1,\tilde{\chi}^2)}[Z_{t+1} = \vct{z}_{t+1}\mid\vct{c}_t,\gamma_{1:t}^{1:2}]}.
\end{align}

Notice that the expression (\ref{b1}) is well-defined, that is, the denominator is non-zero because of our assumption that the realization $c_{t+1},\gamma_{1:t}^{1:2}$ has non-zero probability of occurrence. Let us consider the numerator in the expression (\ref{b1}). For convenience, we will denote it with $\Py^{(\tilde{\chi}^1,\tilde{\chi}^2)}[\vct{x}_{t+1}, \vct{p}_{t+1}^{1:2},\vct{z}_{t+1} \mid \vct{c}_{t},\gamma_{1:t}^{1:2}]$. We have
\begin{align}
\nonumber&\Py^{(\tilde{\chi}^1,\tilde{\chi}^2)}[\vct{x}_{t+1}, \vct{p}_{t+1}^{1:2},\vct{z}_{t+1} \mid \vct{c}_{t},\gamma_{1:t}^{1:2}]\\
&= \sum_{\vct{x}_t,\vct{p}^{1:2}_t,\vct{u}_t^{1:2}}\pi_t(\vct{x}_t,\vct{p}_t^{1:2})\gamma_t^1( \vct{p}_t^1;\vct{u}_t^1)\gamma_t^2( \vct{p}_t^2; {u}_t^2)\Py^{(\tilde{\chi}^1,\tilde{\chi}^2)}[\vct{x}_{t+1}, \vct{p}_{t+1}^{1:2},\vct{z}_{t+1} \mid \vct{c}_{t},\gamma_{1:t}^{1:2},\vct{x}_t,\vct{p}_t^{1:2},\vct{u}_t^{1:2}]\\
&= \label{condindepeq1}\sum_{\vct{x}_t,\vct{p}^{1:2}_t,\vct{u}_t^{1:2}}\pi_t(\vct{x}_t,\vct{p}_t^{1:2})\gamma_t^1(\vct{p}_t^1;u_t^1)\gamma_t^2( \vct{p}_t^2; u_t^2)\Py[\vct{x}_{t+1}, \vct{p}_{t+1}^{1:2},\vct{z}_{t+1} \mid \vct{x}_t,\vct{p}_t^{1:2},\vct{u}_t^{1:2}]\\
&= P_t^j(\pi_t,\gamma_t^{1:2};z_{t+1}, x_{t+1},p_{t+1}^{1:2}),
\end{align}
{where $\pi_t$ is the common information belief on $X_t, P_t^1,P_t^2$ at time $t$ given the realization\footnote{Note that the belief $\Py^{(\tilde{\chi}^1,\tilde{\chi}^2)}[x_t, p_t^{1:2} \mid c_t,\gamma_{1:t-1}^{1:2}] = \Py^{(\tilde{\chi}^1,\tilde{\chi}^2)}[x_t, p_t^{1:2} \mid c_t,\gamma_{1:t}^{1:2}]$ because $\gamma_t^i = \tilde{\chi}_t^i( c_t,\gamma_{1:t-1}^{1:2})$, $i=1,2$. } $c_t,\gamma_{1:t-1}^{1:2}$ and $P_t^j$ is as defined in Definition \ref{fdef}.}
The equality in (\ref{condindepeq1}) is due to the structure of the system dynamics in game $\gm{G}_e$ described by equations (\ref{virdyn1}) -- (\ref{virdyn5}). Similarly, the denominator in (\ref{b1}) satisfies
\begin{align}
0 < \nonumber\Py^{(\tilde{\chi}^1,\tilde{\chi}^2)}[\vct{z}_{t+1} \mid \vct{c}_{t},\gamma_{1:t}^{1:2}] &= \sum_{x_{t+1},p_{t+1}^{1:2}} P_t^j(\pi_t,\gamma_t^{1:2};z_{t+1}, x_{t+1},p_{t+1}^{1:2})\\
&= P_t^m(\pi_t,\gamma_t^{1:2};z_{t+1}), \label{marginalprob}
\end{align}
where $P_t^m$ is as defined is Definition \ref{fdef}. Thus, from equation (\ref{b1}), we have
\begin{equation}\label{beltrans}
\pi_{t+1} = \frac{\vct{P}_t^j(\pi_t,\gamma_t^{1:2};z_{t+1},\cdot,\cdot)}{P_t^m(\pi_t,\gamma_t^{1:2},z_{t+1})} = F_t(\pi_t, \gamma_t^{1:2};\vct{z}_{t+1}),
\end{equation}
where $F_t$ is as defined in Definition \ref{fdef}. {Since the relation (\ref{beltrans}) holds for every realization $c_{t+1}, \gamma_{1:t}^{1:2}$ that has non-zero probability of occurrence under $(\tilde{\chi}^1,\tilde{\chi}^2)$,} we can conclude that  the common information belief $\Pi_t$ evolves \emph{almost surely} as
\begin{align}
\Pi_{t+1} = F_t(\Pi_t, \Gamma_t^{1:2},\vct{Z}_{t+1}), ~~ t \geq 1,
\end{align}
under the strategy profile $(\tilde{\chi}^1,\tilde{\chi}^2)$.

The expected cost at time $t$ can be expressed as follows
\begin{align}
\E^{(\tilde{\vct{\chi}}^1,\tilde{\vct{\chi}}^2)}[c_t(X_t,U_t^1,U_t^2)]
& = \E^{(\tilde{\vct{\chi}}^1,\tilde{\vct{\chi}}^2)}[\E[c_t(X_t,U_t^1,U_t^2) \mid C_t,\Gamma_{1:t}^{1:2}]]\\
&=  \E^{(\tilde{\vct{\chi}}^1,\tilde{\vct{\chi}}^2)}[\tilde{c}_t(\Pi_t,\Gamma_t^1,\Gamma_t^2)],
\end{align}
where the function $\tilde{c}_t$ is as defined in equation (\ref{tildec}).
Therefore, the total cost can be expressed as
\begin{align}
\E^{(\tilde{\vct{\chi}}^1,\tilde{\vct{\chi}}^2)}&\left[\sum_{t=1}^T c_t(\rv{X}_t,\rv{U}_t^{1},\rv{U}_t^2)\right]
 = \E^{(\tilde{\vct{\chi}}^1,\tilde{\vct{\chi}}^2)}\left[\sum_{t=1}^T \tilde{c}_t(\Pi_t,\Gamma_t^1,\Gamma_t^2)\right].
\end{align}

\input{appGH}
\section{Proof of Theorem \ref{dp}}
\label{dpproof}
\input{dpproof.tex}
\input{appJ_v2}



\section{Proof of Lemma \ref{beliefupdateone}}
\label{beliefupdateoneproof}
We construct a transformation $F_t$ as described in Definition \ref{fdef} and show that this transformation  does not use virtual player 2's prescription $\gamma_t^2$.

In Game $\mathscr{H}_e$, the corresponding transformation $P_t^j$ (see Definition \ref{fdef}) has the following form
\begin{align}
P_t^j(\pi_t,\gamma_t^{1:2}; z_{t+1}, x_{t+1}) &= \gamma_t^2(\varnothing; \vct{u}_t^2)\sum_{\vct{x}_t,\vct{u}_t^1}\pi_t(\vct{x}_t)\gamma_t^1( \vct{x}_t ; \vct{u}_t^1)\Py[\vct{x}_{t+1},\vct{y}^2_{t+1} \mid \vct{x}_t,\vct{u}_t^1,\vct{u}_t^2]\\
\label{simple1}& =:  \gamma_t^2(\vct{u}_t^2)Q_t(\pi_t,\gamma_t^1, z_{t+1}; x_{t+1}).
\end{align}
Note that $z_{t+1} = \{y_{t+1}^2,u_t^2\}$. The corresponding transformation $P_t^m$ (see Definition \ref{fdef}) has the following form
\begin{align}
\nonumber P_t^m(\pi_t,\gamma_t^{1:2}; z_{t+1})&= \sum_{\vct{x}_{t+1}}\gamma_t^2(\vct{u}_t^2)Q_t(\pi_t,\gamma_t^1,z_{t+1};x_{t+1})\\
\label{simple2}& =: \gamma_t^2(\vct{u}_t^2)R_t(\pi_t,\gamma_t^1,z_{t+1}).
\end{align}
Following the methodology in Definition \ref{fdef}, we define $F_t$ as
\begin{align}
\label{tempfdef}F_t(\pi_t,\gamma_t^{1:2},z_{t+1})&= 
\begin{cases}
\frac{P_t^j(\pi_t,\gamma_t^{1:2},z_{t+1},\cdot)}{P_t^m(\pi_t,\gamma_t^{1:2};z_{t+1})} &\text{if } P_t^m(\pi_t,\gamma_t^{1:2};z_{t+1}) > 0\\
G_t(\pi_t,\gamma_t^{1:2},z_{t+1}) & \text{otherwise},
\end{cases}
\end{align}
where the transformation $G_t$ is chosen to be
\begin{align}
G_t(\pi_t,\gamma_t^{1:2},z_{t+1})&= 
\begin{cases}
\frac{Q_t(\pi_t,\gamma_t^{1},z_{t+1})}{R_t(\pi_t,\gamma_t^{1},z_{t+1})} &\text{if } R_t(\pi_t,\gamma_t^{1},z_{t+1}) > 0\\
\mathscr{U}(\X_{t+1}) & \text{otherwise},
\end{cases}
\end{align}
where $\mathscr{U}(\X_{t+1})$ is the uniform distribution over the space $\X_{t+1}$. Using the results (\ref{simple1}) and (\ref{simple2}), we can simplify the expression for the transformation $F_t$ in (\ref{tempfdef}) to obtain the following
\begin{align}
F_t(\pi_t,\gamma_t^{1:2},z_{t+1})&= 
\begin{cases}
\frac{Q_t(\pi_t,\gamma_t^{1},z_{t+1})}{R_t(\pi_t,\gamma_t^{1},z_{t+1})} &\text{if } R_t(\pi_t,\gamma_t^{1},z_{t+1}) > 0\\
\mathscr{U}(\X_{t+1}) & \text{otherwise}.
\end{cases}
\end{align}
This concludes the construction of an update rule $F_t$ in the class $\mathscr{B}$ that does not use virtual player 2's prescription $\gamma_t^2$.

\input{appK_v2.tex}

\section{Proof of Theorem \ref{strategy}}
\label{strategyproof}
Let the strategies $\tilde{\chi}^{1*}$ and $\tilde{\chi}^{1*}$ be as defined in Definition \ref{stratdefone}. Note that the strategy $\tilde{\chi}^{1*}$ uses only the common information $c_t$ and player 1's past prescriptions $\gamma_{1:t-1}^1$. Because of this structure, we have
\begin{equation}\label{structeq}
\vct{\chi}^{1*} := \varrho^1(\tilde{\vct{\chi}}^{1*},\tilde{\vct{\chi}}^{2*}) = \varrho^1(\tilde{\vct{\chi}}^{1*},\tilde{\vct{\chi}}^{2})
\end{equation}
for any strategy $\tilde{\vct{\chi}}^{2}$. Let us assume that $\vct{\chi}^{1*} $ is not a minimax strategy for the game $\gm{H}_v$. Therefore, there exists $\vct{\chi}^2 \in \H^2$ such that
\begin{equation}
\mathcal{J}(\vct{\chi}^{1*},\vct{\chi}^{2}) > S(\gm{H}_v).
\end{equation}
Based on the result in (\ref{structeq}) and the fact that $\varrho^2(\tilde{\chi}^1,\chi^2) = \chi^2$ for any $\tilde{\chi}^1 \in \tilde{\mathcal{H}}^1$, we have
\begin{equation}
(\vct{\chi}^{1*},\vct{\chi}^{2}) = \varrho(\tilde{\vct{\chi}}^{1*},{\vct{\chi}}^{2}).
\end{equation}
Hence, using Lemma \ref{evolequi} we have
\begin{equation}
{\mathcal{J}}(\tilde{\vct{\chi}}^{1*},{\vct{\chi}}^{2}) = \mathcal{J}(\vct{\chi}^{1*},\vct{\chi}^{2}) > S(\gm{H}_v) = S(\gm{H}_e).
\end{equation}
This is a contradiction because $\tilde{\vct{\chi}}^{1*}$ is a minimax strategy of $\gm{H}_e$ due to Lemma \ref{valthmone}. Therefore, ${\vct{\chi}}^{1*}$ must be a minimax strategy of game $\gm{H}_v$.

Combining the definitions of the strategy $\tilde{\chi}^{1*}$ in Definition \ref{stratdefone} and the mapping $\varrho^1$ in Definition \ref{def:rho}, we can easily show using a forward inductive argument that at each time $t \leq T$ and for each $c_t$,
\begin{align}
\chi_t^{1*}(\vct{c}_t) = \Xi_t^1(\pi_t).
\end{align}
Here, $\Xi_t^1$ is as defined in the min-max dynamic program in Section \ref{dp:oneside} and $\pi_t$ is computed in a forward inductive manner using the following relation
\begin{align}
\pi_1(x_1) &= \Py[X_1 = x_1 \mid C_1 = c_1] ~ \forall \; x_1,\\
\pi_{\tau + 1} &= F_\tau(\pi_\tau, \Xi_\tau^1(\pi_\tau),z_{\tau+1}), ~  1 \leq \tau < t.
\end{align}
Further, using the approach in Theorem 1 of \cite{nayyar2014common}, we can show that if $\vct{\chi}^{1*}$ is a minimax strategy in the virtual game $\gm{H}_v$ for player 1, then $\vct{g}^{1*}$ defined in Theorem \ref{strategy} is a minimax strategy in the original game $\gm{H}$ for player 1.

%% file: appGH.tex
\section{Some Continuity Results}\label{sec:cont}

In this section, we will state and prove some technical results that will be useful for proving Lemma \ref{equiexistlemma}.

 Let $\mathcal{S}_t$ denote the set of all probability distributions over the finite set $\mathcal{X}_{t} \times \mathcal{P}_{t}^1 \times \mathcal{P}_{t}^2$. Thus, $\mathcal{S}_t$ is the set of all possible common information based beliefs at time $t$.
Define 
\begin{align}
\bar{\mathcal{S}}_t := \{\alpha\pi_t : 0\leq\alpha\leq 1, \pi_t \in \mathcal{S}_t\}.
\end{align} 
The functions $\tilde{c}_t$ in (\ref{tildec}), $P_t^j$ in (\ref{jointprob}), $P_t^m$ in (\ref{marginalprob}) and $F_t$ in (\ref{beltrans}) were defined for any $\pi_t \in \mathcal{S}_t$. We will extend the domain of the argument $\pi_t$ in these functions to $\bar{\mathcal{S}}_t$ as follows. For any {$\gamma_t^i \in \mathcal{B}_t^i, i = 1,2$, $z_{t+1} \in \mathcal{Z}_{t+1}$,} $0 \leq \alpha \leq 1$ and $\pi_t \in {\mathcal{S}}_t$, let
\begin{enumerate}[(i)]
\item $\tilde{c}_t(\alpha\pi_t,\gamma_t^1,\gamma_t^2) := \alpha\tilde{c}_t(\pi_t,\gamma_t^1,\gamma_t^2)$
\item $P^j_t(\alpha\pi_t,\gamma_t^{1:2}) := \alpha P^j_t(\pi_t,\gamma_t^{1:2})$
\item $P^m_t(\alpha\pi_t,\gamma_t^{1:2}) := \alpha P^m_t(\pi_t,\gamma_t^{1:2})$
\item $
F_t(\alpha\pi_t,\gamma_t^{1:2},z_{t+1}) := 
\begin{cases}
F_t(\pi_t,\gamma_t^{1:2},z_{t+1}) &\text{if } \alpha > 0\\
\bm{0} & \text{if } \alpha = 0,
\end{cases}
$
\end{enumerate}
where $\bm{0}$ is a zero-vector of size $|\mathcal{X}_{t} \times \mathcal{P}_{t}^1 \times \mathcal{P}_{t}^2|$.


Having extended the domain of the above  functions, we can also extend the domain of the argument $\pi_t$ in the functions $w_t^u(\cdot), w_t^l(\cdot), V_t^u(\cdot), V_t^l(\cdot)$ defined in the dynamic programs of  Section \ref{dpsec}. First, for any  $0 \leq \alpha \leq 1$ and  $\pi_{T+1} \in {\mathcal{S}}_{T+1}$, define $V^u_{T+1}(\alpha\pi_{T+1}) :=0$.  We can then define the following functions for every  $t \leq T$ in a backward inductive manner:  For any {$\gamma_t^i \in \mathcal{B}_t^i, i = 1,2$, } $0 \leq \alpha \leq 1$ and $\pi_t \in {\mathcal{S}}_t$, let
\begin{align}\label{eq:w_ext}
w^u_t (\alpha\pi_t,\gamma_t^1,\gamma_t^2) &:= \tilde{c}_t(\alpha\pi_t,\gamma_t^1,\gamma_t^2) + \sum_{z_{t+1}}\big[P^m_t(\alpha\pi_t,\gamma_t^{1:2};z_{t+1})V^u_{t+1}(F_t(\alpha\pi_t,\gamma_t^{1:2},z_{t+1}))\big] \\
V^u_t(\alpha\pi_t) &:= \inf_{{\gamma}_t^1} \sup_{\gamma_t^2}w_t^u(\alpha\pi_t,\gamma_t^1,\gamma_t^2).
\end{align}
Note that when $\alpha=1$, the above definition of $w^u_t$ is equal to the definition of $w^u_t$ in equation (\ref{minequa}) of the dynamic program. We can similarly extend $w^l_t$ and $V^l_t$. 
 These extended value functions satisfy the following homogeneity property. A similar result  was shown in \cite[Lemma III.1]{li2014lp} for a special case of our model.


\begin{lemma}\label{scalinglemma}
For any constant $0 \leq \alpha \leq 1$ and any $\pi_t \in \bar{\mathcal{S}}_t$, we have $\alpha V_t^u(\pi_t) = V_t^u(\alpha \pi_t)$ and $\alpha V_t^l(\pi_t) = V_t^l(\alpha \pi_t)$.
\end{lemma}
\begin{proof}
{The proof can be easily obtained from the above definitions of the extended functions. }
\end{proof}
The following lemmas will be used in Appendix \ref{equiexistlemmaproof} to establish some useful properties of the extended functions.

\begin{lemma}\label{equicontlemma1}
Let $V: \bar{\mathcal{S}}_{t+1} \to \R$ be a continuous function satisfying $V(\alpha\pi) = \alpha V(\pi)$ for every $0 \leq \alpha \leq 1$ and $\pi \in \bar{\mathcal{S}}_{t+1}$.
Define 
\begin{align*}
V'(\pi_t, \gamma_t^1,\gamma_t^2) := \sum_{z_{t+1}}P_t^m( \pi_t,\gamma_{t}^{1:2};z_{t+1})[V(F_t(\pi_{t},\gamma_t^{1:2},{z}_{t+1}))]. 
\end{align*}
For a fixed $\gamma_t^1,\gamma_t^2$, $V'(\cdot,\gamma_t^1,\gamma_t^2)$ is a function from $\bar{\mathcal{S}}_{t+1}$ to $\R$.  
Then, the family of functions 
\begin{align}
\mathscr{F}_1&:=\{V'(\cdot,\gamma_t^1,\gamma_t^2): \gamma_t^i \in \mathcal{B}_t^i,i=1,2\}
\end{align}
is equicontinuous. Similarly,  the following families of functions
\begin{align}
\mathscr{F}_2&:=\{V'(\pi_t,\cdot,\gamma_t^2): \gamma_t^2 \in \mathcal{B}_t^2, \pi_t \in \bar{\mathcal{S}}_t\}\\
\mathscr{F}_3&:=\{V'(\pi_t,\gamma_t^1,\cdot): \gamma_t^1 \in \mathcal{B}_t^1, \pi_t \in \bar{\mathcal{S}}_t\}
\end{align}
are  equicontinuous in their respective arguments.
\end{lemma}
\begin{proof}
A continuous function is bounded and uniformly continuous over a compact domain (see Theorem 4.19 in \cite{rudin1964principles}). Therefore, $V$ is bounded and uniformly continuous over $\bar{\mathcal{S}}_{t+1}$.

Using the fact  that $V(\alpha\pi) = \alpha V(\pi)$ and the definition of $F_t$ in Definition \ref{fdef}, the function $V'$ can be written as
\begin{align}\label{vsum}
V'(&\pi_t,\gamma_t^1,\gamma_t^2) = \sum_{\vct{z}_{t+1}}V\left({\vct{P}_t^j(\pi_t,\gamma_t^{1:2};\vct{z}_{t+1},\cdot,\cdot)}\right).
\end{align}
Recall that $P_t^j$ is trilinear in $\pi_t,\gamma_t^1$ and $\gamma_t^2$ with bounded coefficients for a fixed value of $z_{t+1}$ (see (\ref{jointprob})). Therefore, for each $\vct{z}_{t+1}$, $\{P^j_t(\cdot, \gamma^1_t,\gamma^2_t,\vct{z}_{t+1})\}$ is an equicontinuous family of functions in the argument $\pi_t$, where $P^j_t(\pi_t, \gamma^1_t,\gamma^2_t,\vct{z}_{t+1})$ is a short hand notation for the measure $P^j_t(\pi_t, \gamma^1_t,\gamma^2_t,\vct{z}_{t+1},\cdot,\cdot)$ over the space $\X_{t+1} \times \P_{t+1}^1 \times \P_{t+1}^2$. 

Also, since $V$ is uniformly continuous, the family $\left\{V\left({\vct{P}_t^j(\cdot,\gamma_t^{1:2},\vct{z}_{t+1})}\right)\right\}$ is equicontinuous in $\pi_t$ for each $z_{t+1}$. This is because composition with a uniformly continuous function preserves equicontinuity. 
Therefore, the family of functions $\mathscr{F}_1$ is equicontinuous in $\pi_t$. We can use similar arguments to prove equicontinuity of the other two families.
\end{proof}

\begin{lemma}\label{supcont}
{Let $w: \mathcal{B}^1_t \times \mathcal{B}^2_t \to \R$ be a function such that (i) the family of functions $\{w(\cdot,\gamma^2): \gamma^2 \in \mathcal{B}^2_t\}$ is equicontinuous in  the first argument; (ii) the family of functions $\{w(\gamma^1, \cdot): \gamma^1 \in \mathcal{B}^1_t\}$ is equicontinuous in  the second argument. Then $\sup_{\gamma^2}w(\gamma^1,\gamma^2)$ is a  continuous  function of  $\gamma^1$ and, similarly, $\inf_{\gamma^1}w(\gamma^1,\gamma^2)$ is a continuous function of  $\gamma^2$.  }
\end{lemma}
\begin{proof}
Let $\epsilon > 0$. For a given $\gamma^1$, there exists a $\delta > 0$ such that
\begin{align}
|w(\gamma^1,\gamma^2) - w(\gamma'^1,\gamma^2)| \leq \epsilon \quad \forall \gamma^2, \forall ||\gamma^1 - \gamma'^1|| \leq \delta. \label{eq:cont1}
\end{align}
Let $\bar{\gamma}^2$ be a prescription such that  
\begin{equation}\label{eq:sup1}
w(\gamma^1,\bar{\gamma}^2) = \sup_{\gamma^2}w(\gamma^1,\gamma^2).
\end{equation}
Note that the existence of $\bar{\gamma}^2$  is guaranteed because of continuity of $w(\gamma^1, \cdot)$ in the second argument and compactness of $\mathcal{B}^2_t$. Pick any $\gamma'^1$ satisfying $||\gamma^1 - \gamma'^1|| \leq \delta$.  
 Let $\bar{\gamma}'^2$ be a prescription such that  
 \begin{equation}\label{eq:sup2}
 w(\gamma'^1,\bar{\gamma}'^2) = \sup_{\gamma^2}w(\gamma'^1,\gamma^2).
 \end{equation}
   Using \eqref{eq:cont1}, we have
\begin{align}
(i)~~w(\gamma^1,\bar{\gamma}^2) - w(\gamma'^{1},\bar{\gamma}'^2) &\geq w(\gamma^1,\bar{\gamma}'^2) - w(\gamma'^{1},\bar{\gamma}'^2)\notag \\
&\geq  -\epsilon, \label{eq:eps1}\\
(ii)~~w(\gamma^1,\bar{\gamma}^2) - w(\gamma'^{1},\bar{\gamma}'^2) &\leq w(\gamma^1,\bar{\gamma}^2) - w(\gamma'^{1},\bar{\gamma}^2)\notag\\
&\leq \epsilon \label{eq:eps2}.
\end{align}
Equations \eqref{eq:sup1} - \eqref{eq:eps2} imply that 
  $\sup_{\gamma^2}w(\gamma^1,\gamma^2)$ is a continuous function of  $\gamma^1$. We can use a similar argument for showing continuity of  $\inf_{\gamma^1}w(\gamma^1,\gamma^2)$ in $\gamma^2$. 
\end{proof}

\section{Proof of Lemma \ref{equiexistlemma}}\label{equiexistlemmaproof}


We first use the definitions of extensions of $w^u_t,w^l_t,V^u_t,V^l_t$ in Appendix \ref{sec:cont} and Lemmas \ref{scalinglemma} and  \ref{equicontlemma1} to establish the following equicontinuity result.
\begin{lemma}\label{equicontlemma}
The families of functions
\begin{align}
\label{eqcontpi}\mathscr{F}_t^a&:=\{w_t^u(\cdot,\gamma_t^1,\gamma_t^2): \gamma_t^i \in \mathcal{B}_t^i,i=1,2\}\\
\label{eqcont1}\mathscr{F}_t^b&:=\{w_t^u(\pi_t,\cdot,\gamma_t^2): \gamma_t^2 \in \mathcal{B}_t^2, \pi_t \in \bar{\mathcal{S}}_t\}\\
\label{eqcont2}\mathscr{F}_t^c&:=\{w_t^u(\pi_t,\gamma_t^1,\cdot): \gamma_t^1 \in \mathcal{B}_t^1, \pi_t \in \bar{\mathcal{S}}_t\}
\end{align}
are all equicontinuous in their arguments for every $t \leq T$. A similar statement holds for $w_t^l$.
\end{lemma}
\begin{proof}
We use a backward induction argument for the proof.  For induction, assume that $V_{t+1}^u$ is a continuous function for some $t \leq T$. This is clearly true for $t=T$. Using the continuity of $V^u_{t+1}$ we will establish the statement of the lemma for time $t$ and also prove the continuity of $V^u_t$. This establishes the lemma for all $t \leq T$.

\emph{Equicontinuity of $w^u_t$:} Since $\tilde{c}_t(\pi_t,\gamma_t^1,\gamma_t^2)$ is linear in $\pi_t$ with uniformly bounded coefficients for any given $\gamma_t^{1:2}$ (see (\ref{tildec})), it is equicontinuous in the argument $\pi_t$. In Lemma \ref{scalinglemma}, we showed that the value functions $V_t^u$ satisfy the condition $V_t^u(\alpha\pi) = \alpha V_t^u(\pi)$ for every $0 \leq \alpha \leq 1$, $\pi \in \mathcal{S}_{t}$. Further, due to our induction hypothesis, $V_{t+1}^u$ is continuous. Thus, using Lemma \ref{equicontlemma1}, the second term of $w_t^u$,
$$
 \sum_{z_{t+1}}P_t^m( \pi_t,\gamma_{t}^{1:2};z_{t+1})V_{t+1}^u(F_t(\pi_{t},\gamma_t^{1:2},{z}_{t+1})),
$$
is also equicontinuous in $\pi_t$. Hence, the family $\mathscr{F}_t^a$ is equicontinuous in $\pi_t$.

\emph{Continuity of $V^u_t$:}  Due to the equicontinuity of the family $\mathscr{F}_t^a$, we have the following. For any given $\epsilon > 0$ and $\pi_t \in \bar{\mathcal{S}}_t$, there exists a $\delta > 0$ such that
\begin{equation}
|w^u_t(\pi_t,\gamma_t^1,\gamma_t^2)-w^u_t(\pi'_t,\gamma_t^1,\gamma_t^2)| < \epsilon
\end{equation}
for every $\gamma_t^1, \gamma_t^2$ and $\pi_t'$ satisfying $||\pi_t - \pi'_t|| < \delta$. Therefore,
\begin{align}
&w^u_t(\pi_t,\gamma_t^1,\gamma_t^2) < w^u_t(\pi'_t,\gamma_t^1,\gamma_t^2) + \epsilon \; \forall \gamma_t^1,\gamma_t^2\\
\implies &\sup_{\gamma_t^2}w^u_t(\pi_t,\gamma_t^1,\gamma_t^2) \leq \sup_{\gamma_t^2}w^u_t(\pi'_t,\gamma_t^1,\gamma_t^2) + \epsilon \;\forall \gamma_t^1\\
\implies &\nonumber\inf_{\gamma_t^1}\sup_{\gamma_t^2}w^u_t(\pi_t,\gamma_t^1,\gamma_t^2) \leq \inf_{\gamma_t^1}\sup_{\gamma_t^2}w^u_t(\pi'_t,\gamma_t^1,\gamma_t^2) + \epsilon\\
\implies &V^u_t(\pi_t) \leq V^u_t(\pi'_t) + \epsilon,
\end{align}
for every $\pi_t'$ that satisfies $||\pi_t - \pi'_t|| < \delta$. Similarly, $V^u_t(\pi_t) \geq V^u_t(\pi'_t) -\epsilon$ for every $\pi_t'$ that satisfies $||\pi_t - \pi'_t|| < \delta$. Therefore, $V^u_t(\pi_t)$ is continuous at $\pi_t$. 

Hence, by induction, we can say that the family $\mathscr{F}_t^a$ is equicontinuous in $\pi_t$ for every $t \leq T$. We can use similar arguments to prove the equicontinuity of the other families. 
\end{proof}

The continuity of $w^u_t$ established above implies that $\sup_{\gamma_t^2}w_t^u(\pi_t,\gamma_t^1,\gamma_t^2)$ is achieved for every $\pi_t, \gamma^1_t.$ Further,  Lemma \ref{equicontlemma} implies that  $w_t^u$ and $w_t^l$ satisfy the equicontinuity conditions in Lemma \ref{supcont} for any given realization of belief $\pi_t$. Therefore, we can use Lemma \ref{supcont} to argue that $\sup_{\gamma_t^2}w_t^u(\pi_t,\gamma_t^1,\gamma_t^2)$ is continuous in $\gamma_t^1$. And since $\gamma_t^1$ lies in the compact space $\mathcal{B}_t^1$, a minmaximizer exists for the function $w_t^u$. Further, we can use the measurable selection condition (see Condition 3.3.2 in \cite{hernandez2012discrete}) to prove the existence of measurable mapping $\Xi_t^1(\pi_t)$ as defined in Lemma \ref{equiexistlemma}. A similar argument can be made to establish the existence of a maxminimizer and a measurable mapping $\Xi_t^2(\pi_t)$ as defined in Lemma \ref{equiexistlemma}. This concludes the proof of Lemma \ref{equiexistlemma}.

%% file: dpproof.tex
 Let us first define a distribution $\tilde{\Pi}_t$ over the space $\mathcal{X}_t \times \P_t^1 \times \P_t^2$ in the following manner. The distribution $\tilde{\Pi}_t$, given $C_t,\Gamma_{1:t-1}^{1:2}$, is recursively obtained using the following relation
\begin{align}
\tilde{\Pi}_1(x_1,p_1^1,p_1^2) &= \Py[X_1 = x_1,P_1^1 = p_1^1,P_1^2 = p_1^2 \mid C_1 ] ~ \forall\; x_1,p_1^1,p_1^2,\\
\tilde{\Pi}_{\tau+1} &= F_\tau(\tilde{\Pi}_\tau, \Gamma_\tau^{1},\Gamma_\tau^2,\vct{Z}_{\tau+1}), ~~ \tau \geq 1,
\end{align}
where $F_\tau$ is as defined in Definition \ref{fdef} in Appendix \ref{infstateproof}. We refer to this distribution as the strategy-independent common information belief (SI-CIB).

Let $\tilde{\chi}^1 \in \tilde{\mathcal{H}}^1$ be any strategy for virtual player 1 in game $\gm{G}_e$. Consider the problem of obtaining virtual player 2's best response to the strategy $\tilde{\chi}^1$ with respect to the cost $\mathcal{J}(\tilde{\chi}^1 ,\tilde{\chi}^2)$ defined in \eqref{eq:virtualJ}. This problem can be formulated as a Markov decision process (MDP) with common information and prescription history $C_t,\Gamma_{1:t-1}^{1:2}$ as the state. The control action at time $t$ in this MDP is $\Gamma_t^2$, which is selected based on the information $C_t,\Gamma_{1:t-1}^{1:2}$ using strategy $\tilde{\chi}^2 \in \mathcal{H}^2$.
The evolution of the state $C_t,\Gamma_{1:t-1}^{1:2}$ of this MDP is as follows
\begin{align}
\{C_{t+1},\Gamma_{1:t}^{1:2}\} = \{C_t,Z_{t+1},\Gamma_{1:t-1}^{1:2}, \tilde{\chi}^1_t(C_t,\Gamma_{1:t-1}^{1:2}),\Gamma_t^2\},
\end{align}
where 
\begin{align}\label{zstatevol}
\Py^{(\tilde{\vct{\chi}}^1,\tilde{\vct{\chi}}^2)}[Z_{t+1} = z_{t+1} \mid C_t,\Gamma_{1:t-1}^{1:2}, \Gamma_t^2] = P_t^m[\tilde{\Pi}_t,\Gamma_t^1,\Gamma_t^2;z_{t+1}],
\end{align}
almost surely. Here, $\Gamma_t^1 = \tilde{\chi}^1_t(C_t,\Gamma_{1:t-1}^{1:2})$ and the transformation $P_t^m$ is as defined in Definition \ref{fdef} in Appendix \ref{infstateproof}.  Notice that due to Lemma \ref{infstate}, the common information belief $\Pi_t$ associated with any strategy profile ${(\tilde{\vct{\chi}}^1,\tilde{\vct{\chi}}^2)}$ is equal to $\tilde{\Pi_t}$ almost surely. This results in the state evolution equation in \eqref{zstatevol}.
The objective of this MDP is to maximize, for a given $\tilde{\chi}^1$, the following cost
\begin{align}
\E^{(\tilde{\vct{\chi}}^1,\tilde{\vct{\chi}}^2)}\left[\sum_{t=1}^T \tilde{c}_t(\tilde{\Pi}_t,\Gamma_t^1,\Gamma_t^2)\right],
\end{align}
where $\tilde{c}_t$ is as defined in equation (\ref{tildec}). Due to Lemma \ref{infstate}, the total expected cost defined above is equal to the cost ${\mathcal{J}}(\tilde{\vct{\chi}}^1,\tilde{\vct{\chi}}^2)$ defined in \eqref{eq:virtualJ}.

The MDP described above can be solved using the following dynamic program. For every realization of virtual players' information $c_{T+1},\gamma_{1:T}^{1:2}$, define $$V^{\tilde{\chi}^1}_{T+1}(c_{T+1},\gamma_{1:T}^{1:2}) := 0.$$
In a backward inductive manner, for each time $t \leq T$ and each realization $c_{t},\gamma_{1:t-1}^{1:2}$, define
\begin{align}
V^{\tilde{\chi}^1}_{t}(c_{t},\gamma_{1:t-1}^{1:2}) &:= \sup_{\gamma_t^2}[\tilde{c}_t(\tilde{\pi}_t,\gamma_t^1,\gamma_t^2)+ \E[V^{\tilde{\chi}^1}_{t+1}(c_{t},Z_{t+1},\gamma_{1:t}^{1:2}) \mid c_t,\gamma_{1:t}^{1:2}]]\label{thm3e1},
\end{align}
where $\gamma_t^1 = \tilde{\chi}^1_t(c_{t},\gamma_{1:t-1}^{1:2})$ and {$\tilde{\pi}_t$ is the SI-CIB associated with the information $c_{t},\gamma_{1:t-1}^{1:2}$}. Note that the measurable selection condition (see condition 3.3.2 in \cite{hernandez2012discrete}) holds for the dynamic program described above. Thus, the value functions $V^{\tilde{\chi}^1}_{t}(\cdot)$ are measurable and there exists a measurable best-response strategy for player 2 which is a solution to the dynamic program described above. Therefore, we have
\begin{align}
\label{dpcommon}\sup_{\tilde{\chi}^2}\mathcal{J}(\tilde{\chi}^1,\tilde{\chi}^2) = \E V^{\tilde{\chi}^1}_{1}(C_{1}).
\end{align}
\begin{claim}\label{upperclaim}
For any strategy $\tilde{\chi}^1 \in \tilde{\mathcal{H}}^1$ and for any realization of virtual players' information $c_{t},\gamma_{1:t-1}^{1:2}$, we have
\begin{align}
\label{uppervalineq}V^{\tilde{\chi}^1}_{t}(c_{t},\gamma_{1:t-1}^{1:2}) \geq V_t^u(\tilde{\pi}_t),
\end{align}
where $V_t^u$ is as defined in (\ref{minequa}) and $\tilde{\pi}_t$ is the SI-CIB belief associated with the instance $c_{t},\gamma_{1:t-1}^{1:2}$. As a consequence, we have
\begin{align}\label{upineq}
\sup_{\tilde{\chi}^2}\mathcal{J}(\tilde{\chi}^1,\tilde{\chi}^2) \geq \E V^{u}_{1}(\Pi_{1}).
\end{align}
\end{claim}
\begin{proof}
The proof is by backward induction. Clearly, the claim is true at time $t = T+1$. Assume that the claim is true for all times greater than $t$. Then we have
\begin{align*}
V^{\tilde{\chi}^1}_{t}(c_{t},\gamma_{1:t-1}^{1:2}) 
&= \sup_{\gamma_t^2}[\tilde{c}_t(\tilde{\pi}_t,\gamma_t^1,\gamma_t^2) + \E[V^{\tilde{\chi}^1}_{t+1}(c_{t},Z_{t+1},\gamma_{1:t}^{1:2}) \mid c_t,\gamma_{1:t}^{1:2}]]\\
&\geq \sup_{\gamma_t^2}[\tilde{c}_t(\tilde{\pi}_t,\gamma_t^1,\gamma_t^2)
 + \E[V^{u}_{t+1}(F_t(\tilde{\pi}_t,\gamma_t^{1:2},Z_{t+1})) \mid c_t,\gamma_{1:t}^{1:2}]]\\
&\geq V_t^u(\tilde{\pi}_t).
\end{align*}
The first equality follows from the definition in (\ref{thm3e1}) and the inequality after that follows from the induction hypothesis. The last inequality is a consequence of the definition of the value function $V_t^u$. This completes the induction argument.
Further, using Claim \ref{upperclaim} and the result in \eqref{dpcommon}, we have
\begin{align*}
\sup_{\tilde{\chi}^2}\mathcal{J}(\tilde{\chi}^1,\tilde{\chi}^2) = \E V^{\tilde{\chi}^1}_{1}(C_{1}) \geq \E V^{u}_{1}(\tilde{\Pi}_{1}) = \E V^{u}_{1}(\Pi_{1}).
\end{align*}
\end{proof}

We can therefore say that
\begin{align}\label{suone}
S^u(\gm{G}_e) &= \inf_{\tilde{\chi}^1}\sup_{\tilde{\chi}^2}\mathcal{J}(\tilde{\chi}^1,\tilde{\chi}^2) \geq \inf_{\tilde{\chi}^1}\E V^{u}_{1}(\Pi_{1}) = \E V^{u}_{1}(\Pi_{1}).
\end{align}
Further, for the strategy $\tilde{\chi}^{1*}$ defined in Definition \ref{stratdef}, the inequalities (\ref{uppervalineq}) and (\ref{upineq}) hold with equality. We can prove this using an inductive argument similar to the one used to prove Claim \ref{upperclaim}. Therefore, we have
\begin{align}\label{sutwo}
S^u(\gm{G}_e) &= \inf_{\tilde{\chi}^1}\sup_{\tilde{\chi}^2}\mathcal{J}(\tilde{\chi}^1,\tilde{\chi}^2) \leq \sup_{\tilde{\chi}^2}\mathcal{J}(\tilde{\chi}^{1*},\tilde{\chi}^2)  = \E V^{\tilde{\chi}^{1*}}_{1}(C_{1})  = \E V^{u}_{1}(\Pi_{1}).
\end{align}
Combining \eqref{suone} and \eqref{sutwo}, we have
\begin{align*}
S^u(\gm{G}_e) = \E V^{u}_{1}(\Pi_{1}).
\end{align*}
Thus, the inequality in \eqref{sutwo} holds with equality which leads us to the result that the strategy $\tilde{\chi}^{1*}$ is a min-max strategy in game $\gm{G}_e$.
A similar argument can be used to show that
\begin{align*}
S^l(\gm{G}_e) = \E V^{l}_{1}(\Pi_{1}),
\end{align*}
and that the strategy $\tilde{\chi}^{2*}$ defined in Definition \ref{stratdef} is a max-min strategy in game $\gm{G}_e$.

%% file: appJ_v2.tex
\section{Proof of Lemma \ref{structlemma}}\label{structlemmaproof}


We will prove the lemma using the following claim.
\begin{claim}\label{claim:one}
Consider any arbitrary strategy $g^1$ for player 1. Then, there exists a strategy $\bar{g}^1$ for player 1 such that, for each $t$, $\bar{g}^1_t$ is a function  of $X_t$ and $I^2_t$ and 
\[
J(\bar{g}^1,g^2) = J(g^1, g^2), ~~\forall g^2 \in \mathcal{G}^2.
\]
\end{claim}

Suppose that the above claim is true. Let $h^1$ be an equilibrium strategy for player 1 (we know one exists due to Proposition \ref{kuhnremark}). Since we are dealing with a zero-sum game, we know that (a)  $h^1$ achieves the infimum in $\inf_{g^1 \in \mathcal{G}^1}\sup_{g^2 \in \mathcal{G}^2}J(g^1,g^2)$, (b) any strategy achieving the infimum in the above inf-sup problem will be an equilibrium strategy for player 1 \cite{osborne1994course}. 

Due to Claim \ref{claim:one}, there exists a strategy $\bar{h}^1$ for player 1 such that, for each $t$, $\bar{h}^1_t$ is a function only of $X_t$ and $I^2_t$ and 
\[
J(\bar{h}^1,g^2) = J(h^1, g^2),
\]
\emph{for every strategy $g^2 \in \mathcal{G}^2$.} Therefore, we have that 
\begin{align*}
\sup_{g^2 \in \mathcal{G}^2}J(\bar{h}^1,g^2) = \sup_{g^2 \in \mathcal{G}^2}J(h^1,g^2) = \inf_{g^1 \in \mathcal{G}^1}\sup_{g^2 \in \mathcal{G}^2}J(g^1,g^2).
\end{align*} 
Thus, $\bar{h}^1$ is an equilibrium strategy for player 1 that uses only the current state and player 2's information.

\textbf{\emph{Proof of Claim \ref{claim:one}:}} We now proceed to prove Claim \ref{claim:one}. 
Consider any arbitrary strategy $g^1$ for player 1.  Let $\iota_t^2 = \{u_{1:t-1}^2,y_{1:t}^2\}$ be a realization of player 2's information $I_t^2$. Define the distribution $\Psi_t(\iota_t^2)$ over the space $\prod_{\tau = 1}^t(\X_{\tau} \times \U_{\tau}^1)$ as follows:
\begin{align*}
\Psi_t(\iota_t^2; x_{1:t},u_{1:t}^1) := \Py^{g^1,h^2}[X_{1:t},U^1_{1:t} = (x_{1:t},u_{1:t}^1) \mid I_t^2 = \iota_t^2],
\end{align*}
if $\iota_t^2$ is \emph{feasible}, that is $\Py^{g^1,h^2}[I_t^2 = \iota_t^2] > 0$, under the \emph{open-loop} strategy $h^2 = (u^2_{1:t-1})$ for player 2. Otherwise, define $\Psi_t(\iota_t^2; x_{1:t},u_{1:t}^1)$ to be the uniform
 distribution over the space $\prod_{\tau = 1}^t(\X_{\tau} \times \U_{\tau}^1)$.

\begin{lemma}\label{claimJ1}
Let $g^1$ be player 1's strategy and let $g^2$ be an arbitrary strategy for player 2. Then for any realization $x_{1:t},u_{1:t}^1$ of the variables $X_{1:t},U^1_{1:t}$, we have
\begin{align*}
\Py^{g^1,g^2}[X_{1:t},U^1_{1:t} = (x_{1:t},u_{1:t}^1) \mid I_t^2] = \Psi_t(I_t^2; x_{1:t},u_{1:t}^1),
\end{align*}
almost surely.
\end{lemma}
\begin{proof}
{From player 2's perspective, the system evolution can be seen in the following manner. The system state at time $t$ is $S_t = (X_{1:t},U_{1:t}^1,I^2_t)$. Player 2 obtains a partial observation $Y_{t}^2$ of the state at time $t$. Using information $\{Y_{1:t}^2,U_{1:t-1}^2\}$, player 2 then selects an action $U_t^2$. The state then evolves in a controlled Markovian manner (with dynamics that depend on $g^1$). Thus, from player 2's perspective,  this system is a partially observable Markov decision process (POMDP). The claim then follows from the standard result in POMDPs that the belief on the state given the player's information does not depend on the player's strategy \cite{kumar2015stochastic}.}\qed
\end{proof}

For any instance $\iota_t^2$ of player 2's information $I_t^2$, define the distribution $\Phi_t(\iota_t^2)$ over the space $\X_t \times \U_t^1$ as follows
\begin{align}
\Phi_t(\iota_t^2; x_t, u_t^1) = {\sum_{x_{1:t-1},u_{1:t-1}^1}\Psi_t(\iota_t^2;x_{1:t},u_{1:t}^1)}.\label{eq:phi}
\end{align}
Define strategy $\bar{g}^1$ for player 1 such that for any realization $x_t, \iota_t^2$ of state $X_t$ and player 2's information $I_t^2$ at time $t$, the probability of selecting an action $u_t^1$ at time $t$ is
\begin{align}\label{eq:defgbar}
\bar{g}_t^1(x_t,\iota_{t}^2;u_t^1)= 
\begin{cases}
\frac{\Phi_t(\iota_t^2; x_t,u_t^1)}{\sum_{u_t^{1'}}\Phi_t(\iota_t^2; x_t,u_t^{1'})} & \text{if } \sum_{u_t^{1'}}\Phi_t(\iota_t^2; x_t,u_t^{1'}) > 0\\
\mathscr{U}(\cdot) & \text{otherwise},
\end{cases}
\end{align}
where $\mathscr{U}(\cdot)$ denotes the uniform distribution over the action space $\U_t^1$. Notice that the construction of the strategy $\bar{g}^1$ does not involve player 2's strategy.

\begin{lemma}\label{aseqstrat}
For any strategy $g^2$ for player 2, we have
\begin{align*}
\Py^{(g^1,g^2)}[U_t^1 = u_t^1 \mid X_t, I_t^2] = \bar{g}_t^1(X_t,I_{t}^2;u_t^1)
\end{align*}
with almost surely for every $u_t^1 \in \U_t^1$.
\end{lemma}
\begin{proof}
Let $x_t, \iota_t^2$ be a realization that has a non-zero probability of occurrence under the strategy profile $(g^1,g^2)$. Then using Lemma \ref{claimJ1}, we have
\begin{align}
\Py^{(g^1,g^2)}[X_{1:t},U^1_{1:t} = (x_{1:t},u_{1:t}^1) \mid \iota_t^2] = \Psi_t(\iota_t^2; x_{1:t},u_{1:t}^1), \label{eq:lemma14}
\end{align}
for every realization $x_{1:t-1}$ of states $X_{1:t-1}$ and $u_{1:t}^1$ of actions $U_{1:t}^1$. Summing over all $x_{1:t-1},u^1_{1:t}$ and using \eqref{eq:phi} and \eqref{eq:lemma14},  we have
\begin{align}
 \Py^{(g^1,g^2)}[X_t = x_t \mid I_t^2 = \iota_t^2] =\sum_{u_t^{1}}\Phi_t(\iota_t^2; x_t,u_t^{1}).
\end{align}
The left hand side of the above equation is positive since $x_t,i^2_t$ is a realization of positive probability under the strategy profile $(g^1,g^2)$.

Using Bayes' rule, \eqref{eq:lemma14}, \eqref{eq:phi} and \eqref{eq:defgbar}, we obtain
\begin{align}
\nonumber\Py^{g^1,g^2}[U_t^1 = u_t^1 \mid X_t = x_t, I_t^2 = \iota_t^2]
&= \frac{\Phi_t(\iota_t^2; x_t,u_t^1)}{\sum_{u_t^{1'}}\Phi_t(\iota_t^2; x_t,u_t^{1'})}= \bar{g}_t^1(x_t,\iota_{t}^2;u_t^1) .
\end{align}
This concludes the proof of the lemma.\qed
\end{proof}

We can now show that the strategy $\bar{g}^1$ defined in \eqref{eq:defgbar} satisfies
\[
J(\bar{g}^1,g^2) = J(g^1, g^2),
\]
for every strategy $g^2 \in \mathcal{G}^2$. 
Because of the structure of the cost function in \eqref{eq:totalcost}, it is sufficient to show that for each time $t$, the random variables   $(X_t,U_t^1,U_t^2,I^2_t)$ have the same joint distribution under strategy profiles $(g^1,g^2)$ and $(\bar{g}^1,g^2)$. We prove this by induction. It is easy to verify that at time $t=1$, $(X_1,U_1^1,U_1^2,I^2_1)$ have the same joint distribution under strategy profiles $(g^1,g^2)$ and $(\bar{g}^1,g^2)$.

Now assume that at time $t$, 
\begin{align}
\label{jointeq}\Py^{g^1,g^2}[x_t,u_t^1,u_t^2,\iota_t^2] = \Py^{\bar{g}^1,g^2}[x_t,u_t^1,u_t^2,\iota_t^2],
\end{align}
 for any realization of state, actions and player 2's information $x_t,u_t^1,u_t^2, \iota_t^2$. Let $\iota_{t+1}^2 = (\iota_t^2,u_t^2,y_{t+1}^2)$. Then we have
\begin{align}
\Py^{g^1,g^2}[x_{t+1},\iota_{t+1}^2] &= \nonumber\sum_{\bar{x}_t,\bar{u}_t^1}\Py[x_{t+1},y_{t+1}^2 \mid \bar{x}_t,\bar{u}_t^1,u_t^2,\iota_t^2]\Py^{g^1,g^2}[\bar{x}_t,\bar{u}_t^1,u_t^2,\iota_t^2]\\
&= \label{indhyp}\sum_{\bar{x}_t,\bar{u}_t^1}\Py[x_{t+1},y_{t+1}^2 \mid \bar{x}_t,\bar{u}_t^1,u_t^2,\iota_t^2]\Py^{\bar{g}^1,g^2}[\bar{x}_t,\bar{u}_t^1,u_t^2,\iota_t^2]\\
&=\Py^{\bar{g}^1,g^2}[x_{t+1},\iota_{t+1}^2].\label{indhyp2}
\end{align}
The equality in (\ref{indhyp}) is due to the induction hypothesis. Note that the conditional distribution $\Py[x_{t+1},\iota_{t+1}^2 \mid {x}_t,{u}_t^1,u_t^2,\iota_t^2]$ does not depend on players' strategies (see equations (\ref{statevol}) and (\ref{obseq})).
 
At $t+1$, for any realization $x_{t+1},u_{t+1}^1,u_{t+1}^2,\iota_{t+1}^2$ that has non-zero probability of occurrence under the strategy profile $(g^1,g^2)$, we have
\begin{align}
\label{constarg}&\Py^{g^1,g^2}[x_{t+1},u_{t+1}^1,u_{t+1}^2,\iota_{t+1}^2] \\
&= \nonumber\Py^{g^1,g^2}[x_{t+1},\iota_{t+1}^2]g_t^2(\iota_{t+1}^2;u_{t+1}^2)\Py^{g^1,g^2}[u_{t+1}^1\mid x_{t+1},\iota_{t+1}^2]\\
&= \Py^{g^1,g^2}[x_{t+1},\iota_{t+1}^2]g_t^2(\iota_{t+1}^2;u_{t+1}^2)\bar{g}_t^1(x_{t+1},\iota_{t+1}^2;u_{t+1}^1)\label{constarg1}\\
&= \Py^{\bar{g}^1,g^2}[x_{t+1},\iota_{t+1}^2]g_t^2(\iota_{t+1}^2;u_{t+1}^2)\bar{g}_t^1(x_{t+1},\iota_{t+1}^2;u_{t+1}^1)\label{constarg3}\\
&= \Py^{\bar{g}^1,g^2}[x_{t+1},\iota_{t+1}^2]g_t^2(\iota_{t+1}^2;u_{t+1}^2)\Py^{\bar{g}^1,g^2}[u_{t+1}^1\mid x_{t+1},\iota_{t+1}^2]\label{constarg2}\\
&= \Py^{\bar{g}^1,g^2}[x_{t+1},u_{t+1}^1,u_{t+1}^2,\iota_{t+1}^2]\label{constarg4},
\end{align}
where the equality in \eqref{constarg1} follows from Lemma \ref{aseqstrat} and the equality in \eqref{constarg3} follows from the result in \eqref{indhyp2}.
Therefore, by induction, the equality in \eqref{jointeq} holds for all $t$. This concludes the proof of Claim \ref{claim:one}. \qed

%% file: appK_v2.tex
\section{Proof of Lemma \ref{piecelemma}}\label{piecelemmaproof}

{We prove the lemma for the extensions of min-max and max-min value functions defined over $\bar{\mathcal{S}}_t$. Recall that $\bar{\mathcal{S}}_t$ is the set of all vectors $\alpha\pi_t$ where $0 \leq \alpha \leq 1$ and $\pi_t$ is a probability distribution on $\mathcal{X}_t \times \mathcal{P}^1_t \times \mathcal{P}^2_t$. We refer the reader to Appendix \ref{sec:cont} for the definition and a key property (Lemma \ref{scalinglemma}) of the extended min-max and max-min value functions. }

\emph{Induction hypothesis at time $t$: } Assume that for every $\tau$, such that $T \geq \tau \geq  t$, (i) the value functions $V^u_{\tau+1}(\pi_{\tau +1}) = V^l_{\tau + 1}(\pi_{\tau +1}) =: V_{\tau + 1}(\pi_{\tau +1})$ for every $\pi_{\tau +1} \in \bar{\mathcal{S}}_{\tau + 1}$, and (ii) $V_{\tau + 1}$ is piecewise linear and convex in $\pi_{\tau +1}$, i.e.
\begin{align}
V_{\tau + 1}(\pi_{\tau +1}) = \max_{\ell \in \mathcal{A}_{\tau +1}}\langle \ell, \pi_{\tau+1} \rangle,
\end{align}
where $\mathcal{A}_{\tau + 1}$ is a finite collection of vectors of size $|\X_{\tau + 1}|$. Note that this hypothesis is true for $t = T$ since $V_{T+1}^u(\pi_{T+1}) = V_{T+1}^l(\pi_{T+1}) =  0 =:V_{T+1}(\pi_{T+1})$ for every $\pi_{T+1}$ by definition. 

Using the induction hypothesis at time $t$, we will prove it for time $t-1$. Thus, we need to establish that (i) $V^u_{t}(\pi_{t}) = V^l_{t}(\pi_{t}) =: V_{t}(\pi_{t})$ for every $\pi_{t} \in \bar{\mathcal{S}}_{t}$, and that  (ii) $V_{t}$ is piecewise linear and convex in $\pi_{t}$.

\subsection{Equality of $V^u_t$ and $V^l_t$}
Consider any $\pi_{t} \in \bar{\mathcal{S}}_{t}$. Since the functions $V_{t+1}^u$ and $V_{t+1}^l$ are identical (due to the induction hypothesis), the functions $w_t^u$ and $w_t^l$ defined in Section \ref{dpsec} (see equation \eqref{eq:w_ext} in Appendix \ref{sec:cont}) are also identical. Let $w_t := w_t^u = w_t^l$. The value functions $V_t^u$ and $V_t^l$ at time $t$ are thus given by
\begin{align}
V_t^u(\pi_t) &:= \inf_{{\gamma}_t^1}\sup_{\gamma_t^2}w_t(\pi_t,\gamma_t^1,\gamma_t^2)\\
V_t^l(\pi_t) &:= \sup_{\gamma_t^2}\inf_{{\gamma}_t^1}w_t(\pi_t,\gamma_t^1,\gamma_t^2).
\end{align}
Note that $V_t^u(\pi_t)$ and $V_t^l(\pi_t)$ can be viewed as the upper and lower values of a single stage zero-sum game in which the minimizing player selects a prescription $\gamma_t^1 \in \mathcal{B}_t^1$, the maximizing player selects a prescription $\gamma_t^2 \in \mathcal{B}_t^2$, and the cost function is $w_t(\pi_t,\gamma_t^1,\gamma_t^2)$. Let us denote this single stage game by $SG_t(\pi_t)$.
According to Sion's minimax theorem \cite{sion1958general},  the upper and lower values of game  $SG_t(\pi_t)$ are equal, i.e. $V_t^u(\pi_t) = V_t^l(\pi_t)$, if the cost $w_t(\pi_t,\gamma_t^1,\gamma_t^2)$ and the prescription spaces $\mathcal{B}_t^1$ and $\mathcal{B}_t^2$ satisfy the following conditions:
(i) The sets $\mathcal{B}_t^1$ and $\mathcal{B}_t^2$ are convex and compact in a finite-dimensional Euclidean space.
(ii) The function $w_t(\pi_t,\gamma_t^1,\gamma_t^2)$ is continuous and convex in $\gamma_t^1$ for any given $\gamma_t^2$, and  it is continuous and concave in $\gamma_t^2$ for any given $\gamma_t^1$. It is clear that the first condition holds for $SG_t(\pi_t)$ because the prescriptions are mappings from a finite space to a probability simplex.  We now focus on the second condition. 


From Lemma \ref{equicontlemma} in Appendix \ref{equiexistlemmaproof}, we know that  $w_t(\pi_t, \gamma_t^1, \gamma_t^2)$ is continuous in $\gamma_t^1$ for any fixed $\gamma_t^2$ and it is continuous in $\gamma_t^2$ for any given $\gamma_t^1$. 
Further, from Lemma \ref{beliefupdateone} in Section \ref{beliefoneside}, we know that the belief update function $F_t$ does not depend on the prescription $\gamma_t^2$. Thus, using the notation defined in equations (\ref{simple1}) and (\ref{simple2}) of Appendix \ref{beliefupdateoneproof}, the function $w_t$, as defined in \eqref{eq:w_ext}, can be expressed as
\begin{align}
\nonumber w_t(\pi_t,\gamma_t^1,\gamma_t^2) &= \tilde{c}_t(\pi_t,\gamma_t^1,\gamma_t^2) + \sum_{\vct{z}_{t+1}}\gamma_t^2(\vct{u}_t^2)R_t(\pi_t,\gamma_t^1,\vct{z}_{t+1})V_{t+1}\left(F_t(\pi_t,\gamma_t^1,z_{t+1})\right)\\
&= \tilde{c}_t(\pi_t,\gamma_t^1,\gamma_t^2) + \sum_{\vct{z}_{t+1}}\gamma_t^2(\vct{u}_t^2)V_{t+1}\left({\vct{Q}_{t}(\pi_t,\gamma_t^1,\vct{z}_{t+1})}\right),\label{wtdef}
\end{align}
where the last equality is because of Lemma \ref{scalinglemma} in Appendix \ref{sec:cont}. Convexity (resp. concavity) condition holds for the first term in \eqref{wtdef} because of its bilinear structure  in prescriptions for a fixed $\pi_t$ (see \eqref{tildec}). Thus, we only need to show that the convexity (resp. concavity) condition holds for the second term in \eqref{wtdef}.

For a fixed $\gamma_t^1$, the second term in \eqref{wtdef} is linear in $\gamma_t^2$. Thus, it is concave in $\gamma_t^2$. Recall that the measure $\vct{Q}_{t}(\pi_t,\gamma_t^1,\vct{z}_{t+1})$ is linear in $\gamma_t^1$ for a fixed $\pi_t$ (see \eqref{simple1}) and $V_{t+1}$ is convex due to our \emph{induction hypothesis}. Therefore, for a fixed $\gamma_t^2$, the composition $V_{t+1}\left({\vct{Q}_{t}(\pi_t,\gamma_t^1,z_{t+1})}\right) $ is convex in $\gamma_t^1$. Hence, the second term in \eqref{wtdef} is convex in $\gamma_t^1$ for a fixed $\gamma_t^2$. Thus, both conditions of Sion's theorem are valid for the single state game $SG_t(\pi_t)$. Hence,  $V^u_{t}(\pi_t) = V^l_{t}(\pi_t) =: V_{t}(\pi_t)$ for every $\pi_t \in \bar{\mathcal{S}}_{t}$.

\subsection{Piecewise Linearity and Convexity of $V_t$}\label{pieceproof}

We  first prove the following claims.
\begin{claim}\label{piececlaim1}
For each $u^2 \in \U_t^2$, there exists a finite set $\mathcal{D}(u^2)$ of vectors of size $|\mathcal{X}_t \times \mathcal{U}_t^1|$ such that
\begin{align}
w_t(\pi_t, \gamma_t^1, \gamma_t^2) &= \sum_{u^2}\gamma_t^2(u^2)\max_{d \in \mathcal{D}(u^2)}\sum_{x_t,u_t^1}d(x_t,u_t^1)\pi_t(x_t)\gamma_t^1(x_t;u_t^1)\\
&=: \sum_{u^2}\gamma_t^2(u^2)\max_{d \in \mathcal{D}(u^2)}\langle d, \pi_t\gamma_t^1 \rangle,
\end{align}
{where $\pi_t\gamma_t^1$ denotes the vector representation of the measure $\pi_t(x_t)\gamma_t^1(x_t;u_t^1)$ over the space $\mathcal{X}_t \times \mathcal{U}_t^1$.}
\end{claim}
\begin{proof}
Consider the second term of $w_t$ in \eqref{wtdef}. Due to the induction hypothesis, we have
\begin{align}
&\sum_{\vct{z}_{t+1}}\gamma_t^2(\vct{u}_t^2)V_{t + 1}(\vct{Q}_{t}(\pi_t,\gamma_t^1,z_{t+1})) \label{secondtermeq}\\
 &= \sum_{\vct{z}_{t+1}}\gamma_t^2(\vct{u}_t^2)\max_{\ell \in \mathcal{A}_{t +1}}\langle \ell, \vct{Q}_{t}(\pi_t,\gamma_t^1,z_{t+1}) \rangle\\
&= \sum_{u_t^2}\gamma_t^2(\vct{u}_t^2)\left[\sum_{y_{t+1}^2}\max_{\ell \in \mathcal{A}_{t +1}}\sum_{x_{t+1}} \ell(x_{t+1})\vct{Q}_{t}(\pi_t,\gamma_t^1,z_{t+1}; x_{t+1})\right].
\end{align}
Note that $z_{t+1} = \{u_t^2,y_{t+1}^2\}$. Using the fact that $Q_t$ is linear in the {product measure} $\pi_t \gamma_t^1$ (see \eqref{simple1}), we can say that the term $$\max_{\ell \in \mathcal{A}_{t +1}}\sum_{x_{t+1}} \ell(x_{t+1})\vct{Q}_{t}(\pi_t,\gamma_t^1,z_{t+1}; x_{t+1})$$ is piecewise linear and convex in the product measure $\pi_t \gamma_t^1$, for each $z_{t+1}$. Further, since the sum of piecewise linear and convex functions is piecewise linear and convex, we can say that for each $u^2 \in \mathcal{U}_t^2$, there exists a finite set $\mathcal{D}'(u^2)$ of vectors of size $|\mathcal{X}_t \times \mathcal{U}_t^1|$ such that {\eqref{secondtermeq} is equal to}
\begin{align}
 \sum_{u^2}\gamma_t^2(u^2)\left[\max_{d \in \mathcal{D}'(u^2)}\sum_{x_t,u_t^1} d(x_t,u_t^1)\pi_t(x_t)\gamma_t^1(x_t;u_t^1)\right]= \sum_{u^2}\gamma_t^2(u^2)\max_{d \in \mathcal{D}'(u^2)}\langle d, \pi_t\gamma_t^1 \rangle.
\end{align}
Combining this fact with \eqref{wtdef}, we have
\begin{align}
w_t(\pi_t, \gamma_t^1, \gamma_t^2) = \tilde{c}_t(\pi_t,\gamma_t^1,\gamma_t^2) + \sum_{u^2}\gamma_t^2(u^2)\max_{d \in \mathcal{D}'(u^2)}\langle d, \pi_t\gamma_t^1 \rangle.\label{wstruct}
\end{align}
Instantaneous cost $\tilde{c}_t$ has the following structure (see \eqref{tildec})
\begin{align*}
\tilde{c}_t(\pi_t,\gamma_t^1,\gamma_t^2) = \sum_{x_t,u_t^1,u_t^2}\pi_t(x_t)\gamma_t^1(x_t;u_t^1)\gamma_t^2(u_t^2)c_t(x_t,u_t^1,u_t^2).
\end{align*}
Thus, using the structure in \eqref{wstruct}, we can conclude that for each $u^2 \in \mathcal{U}_t^2$, there exists a finite set $\mathcal{D}(u^2)$ of vectors of size $|\mathcal{X}_t \times \mathcal{U}_t^1|$ such that
\begin{align}
w_t(\pi_t, \gamma_t^1, \gamma_t^2) = \sum_{u^2}\gamma_t^2(u^2)\max_{d \in \mathcal{D}(u^2)}\langle d, \pi_t\gamma_t^1 \rangle.
\end{align}
This concludes the proof of the claim.\qed
\end{proof}

Let $\mathcal{E}_t$ be the collection of all mappings $e: \U_t^2 \rightarrow \Delta [\cup_{u^2 \in \mathcal{U}_t^2} \mathcal{D}(u^2)]$ such that for every $u^2$, $e(u^2;d) = 0$ if $d \notin \mathcal{D}(u^2)$ or in other words, the support of the distribution $e(u^2)$ is a subset of $\mathcal{D}(u^2)$. For some mapping $e$ and some action $u^2 \in \U_t^2$, the probability associated with a vector $d \in \mathcal{D}(u^2)$ is denoted by $e(u^2;d)$.


\begin{claim}\label{maxmamxminclaim}
We have
\begin{align}
V_t(\pi_t)&=\max_{\gamma_t^2}\min_{\gamma_t^1}\max_{e\in \mathcal{E}_t}\sum_{u^2}\sum_{d \in \mathcal{D}(u^2)}\gamma_t^2(u^2)e(u^2;d)\langle d,\pi_t\gamma_t^1\rangle\\
&= \max_{\gamma_t^2}\max_{e \in \mathcal{E}_t}\min_{\gamma_t^1}\sum_{u^2}\sum_{d \in \mathcal{D}(u^2)}\gamma_t^2(u^2)e(u^2;d)\langle d,\pi_t\gamma_t^1\rangle\label{maxmaxmin}.
\end{align}
\end{claim}
\begin{proof}
Because of Claim \ref{piececlaim1}, we have
\begin{align}
w_t(\pi_t, \gamma_t^1, \gamma_t^2) &= \sum_{u^2}\gamma_t^2(u^2)\max_{d \in \mathcal{D}(u^2)}\langle d, \pi_t\gamma_t^1 \rangle\\
&=\sum_{u^2}\gamma_t^2(u^2)\max_{\lambda \in \Delta \mathcal{D}(u^2)}\sum_{d \in \mathcal{D}(u^2)}\lambda(d)\langle d, \pi_t\gamma_t^1 \rangle\label{deteqrand}\\
&=\max_{e \in \mathcal{E}_t}\sum_{u^2}\sum_{d \in \mathcal{D}(u^2)}\gamma_t^2(u^2)e(u^2;d)\langle d, \pi_t\gamma_t^1 \rangle.
\end{align}
For any fixed $\gamma_t^2$, we have
\begin{align}
\min_{\gamma_t^1}w_t(\pi_t, \gamma_t^1, \gamma_t^2) &= \min_{\gamma_t^1}\max_{e \in \mathcal{E}_t}\sum_{u^2}\sum_{d \in \mathcal{D}(u^2)}\gamma_t^2(u^2)e(u^2;d)\langle d, \pi_t\gamma_t^1 \rangle\\
&=\max_{e \in \mathcal{E}_t}\min_{\gamma_t^1}\sum_{u^2}\sum_{d \in \mathcal{D}(u^2)}\gamma_t^2(u^2)e(u^2;d)\langle d, \pi_t\gamma_t^1 \rangle.\label{maxmineq}
\end{align}
The last inequality is a consequence of Sion's minimax theorem \cite{sion1958general}. We can apply Sion's theorem here because the sets $\mathcal{B}_t^1$ and $\mathcal{E}_t$ are compact and convex and, the function
$$\sum_{u^2}\sum_{d \in \mathcal{D}(u^2)}\gamma_t^2(u^2)e(u^2;d)\langle d, \pi_t\gamma_t^1 \rangle$$
is convex in $\gamma_t^1$ and concave in $e$. The claim then follows from \eqref{maxmineq} and the fact that
\begin{align}
V_t(\pi_t) = \max_{\gamma_t^2}\min_{\gamma_t^1}w_t(\pi_t, \gamma_t^1, \gamma_t^2).
\end{align}\qed
\end{proof}

Let $\mathscr{P}_t \subset \R^{|\U_t^2 \times \cup_{u^2 \in \U_t^2}\mathcal{D}(u^2)|}$ be the polytope that is characterized by the following constraints 
\begin{align}
&\sum_{u^2\in \U_t^2,d \in \cup_{u^2 \in \U_t^2}\mathcal{D}(u^2)}r(u^2,d) = 1\\
& r(u^2,d) \geq 0\; \forall u^2 \in \U_t^2, d \in \cup_{u^2 \in \U_t^2}\mathcal{D}(u^2)\\
& r(u^2,d) = 0 \; \text{if } d \notin \mathcal{D}(u^2).
\end{align}
Notice that the objective in the max-max-min problem in \eqref{maxmaxmin}, is a function of the product $\gamma_t^2(u^2) e(u^2,d)$. The product $\gamma_t^2e$ is a joint distribution over the space $\U_t^2 \times \cup_{u^2 \in \U_t^2}\mathcal{D}(u^2)$. One can easily show that the set
\begin{align}
\mathscr{P}_t = \{\gamma_t^2 e : \gamma_t^2 \in \mathcal{B}_t^2, e \in \mathcal{E}_t\}.
\end{align}
Therefore, using Claim \ref{maxmamxminclaim}, we have
\begin{align}
V_t(\pi_t) &= \max_{\gamma_t^2}\max_{e \in \mathcal{E}_t}\min_{\gamma_t^1}\sum_{u^2}\sum_{d \in \mathcal{D}(u^2)}\gamma_t^2(u^2)e(u^2;d)\langle d,\pi_t\gamma_t^1\rangle\label{maxmaxmin2}\\
&= \max_{r \in \mathscr{P}_t}\min_{\gamma_t^1}\sum_{u^2}\sum_{d \in \mathcal{D}(u^2)}r(u^2,d)\langle d,\pi_t\gamma_t^1\rangle.\label{maxmaxmin3}
\end{align}

\begin{claim}\label{lpclaim}
The value of the expression (\ref{maxmaxmin}), and thus $V_t(\pi_t)$, is equal to the optimal value of the following linear program
\begin{equation*}
\begin{aligned}
& \underset{r \in \mathscr{P}_t, \nu \in \R^{|\X_t|}}{\text{max}}
 \sum_{x_t}\pi_t(x_t)\nu(x_t) \\
& \text{s. t. }
 \sum_{u^2 \in \U_t^2, d \in \cup_{u^2 \in \U_t^2}\mathcal{D}(u^2)}r(u^2,d)d({x_t,u_t^1}) \geq \nu(x_t), \; \forall u_t^1 \in \mathcal{U}_t^1,x_t \in \mathcal{X}_t.
\end{aligned}
\end{equation*}
\end{claim}

\begin{proof}
Consider the following imaginary Bayesian zero-sum game. Fix $\pi_t$. Nature selects $x_t$ using the distribution $\pi_t$. The minimizing player in this imaginary game observes $x_t$ and selects a distribution over the space $\U_t^1$ using the prescription $\gamma_t^1$. The maximizing player selects a distribution from the set $\mathscr{P}_t$. The cost function associated with these selections is
\begin{align}
p(\gamma_t^1,r) &:= \sum_{u^2}\sum_{d \in \mathcal{D}(u^2)}\sum_{x_t,u_t^1}r(u^2,d) \pi_t(x_t)\gamma_t^1(x_t;u_t^1)d(x_t,u_t^1)\\
&= \sum_{u^2}\sum_{d \in \mathcal{D}(u^2)}r(u^2,d)\langle d, \pi_t\gamma_t^1 \rangle.
\end{align}
Notice that the value of this finite Bayesian game exists and is equal to the maxmin value in \eqref{maxmaxmin3} and thus, equal to the value function $V_t(\pi_t)$. The value of this game can be obtained by the linear program given in \cite{ponssard1980lp} which leads us to the desired result.\qed
\end{proof}

As mentioned earlier, the optimum value of the linear program in Claim \ref{lpclaim} is finite. Thus, there exists a solution to this linear program such that it is an extreme point of the polytope formed by the constraints (see Corollary 2 in Chapeter 2 of \cite{luenberger}). Also, since there are only finitely many extreme points (see Corollary 3 in Chapeter 2 of \cite{luenberger}), we can restrict the optimization variables in the linear program to this finite set of extreme points. Let this set of extreme points be $\mathscr{E}_t$. Therefore,
\begin{align}
V_t(\pi_t) = \max_{r,\nu \in \mathscr{E}_t}\sum_{x_t}\pi_t(x_t)\nu(x_t).
\end{align}
This proves that the value function $V_t(\pi_t)$ is piecewise linear and convex in $\pi_t$. This completes the induction step of our proof.